\documentclass{llncs}  

\usepackage{booktabs}

\usepackage{amsmath}
\usepackage{amssymb}
\usepackage{stmaryrd}
\usepackage{xstring}
\usepackage{hyperref}

\usepackage{url}

\usepackage{graphicx}
\usepackage{tikz}
\usepackage{color}
\usepackage{verbatim}
\usepackage{enumerate}

\newcommand{\mycond}[2]{\mathit{In}^{#1}_{#2}(\CSequence,\CSequenced,\Call)}
\newcommand{\myconclin}[2]{\mathit{Concl}^{#1}_{#2}(\CSequence,\CSequenced,\Call)}
\newcommand{\myconclout}[2]{\mathit{Concl}^{#1}_{#2}(\CSequence,\CSequenced,\Call,\Calld)}

\newcommand{\pie}[1]{%
\begin{tikzpicture}[scale=0.7]
 \draw (0,0) circle (1ex);\fill (1ex,0) arc (0:#1:1ex) -- (0,0) -- cycle;
\end{tikzpicture}%
}

\newcommand{\pieup}[1]{%
\begin{tikzpicture}[scale = 0.5]
 \draw (0,0) circle (1ex);\fill (1ex,0) arc (0:#1:1ex) -- (0,0) -- cycle;
\end{tikzpicture}%
}

\newcommand{\myprivacy}[1]{%
    \IfEqCase{#1}{%
        {1}{\ensuremath{\scalebox{0.8}{\pie{0}}}}%
        {2}{\ensuremath{\scalebox{0.8}{\pie{180}}}}%
        {3}{\ensuremath{\scalebox{0.8}{\pie{360}}}}%
    }[\PackageError{privacy}{Undefined option to privacy: #1}{}]%
}%

\newcommand{\myprivacyup}[1]{%
    \IfEqCase{#1}{%
        {1}{\ensuremath{\pieup{0}}}%
        {2}{\ensuremath{\pieup{180}}}%
        {3}{\ensuremath{\pieup{360}}}%
    }[\PackageError{privacyup}{Undefined option to privacy: #1}{}]%
}%

\newcommand{\observance}[1]{%
    \IfEqCase{#1}{%
        {before}{\ensuremath{\beta}}
        {after}{\ensuremath{\alpha}}
    }[\PackageError{openness}{Undefined option to openness: #1}{}]%
}%

\newcommand{\prtype}{\mathsf{p}}
\newcommand{\prset}{\mathsf{P}}
\newcommand{\dirtype}{\mathsf{d}}
\newcommand{\dirset}{\mathsf{D}}
\newcommand{\otype}{\mathsf{o}}
\newcommand{\oset}{\mathsf{O}}

\newcommand{\callupone}[2]{#1^{\,#2}}

\newcommand{\langu}{\Formulas}
\newcommand{\lang}{\langu}

\newcommand{\cc}{\CSequence}

\newcommand{\calltype}{\tau}

\newcommand{\outcond}[3]{\mathit{Out}^{#1}_{#2}(\Call,\Calld)}

\newcommand{\NI}{\noindent}
\newcommand{\bfe}[1]{\begin{bfseries}\emph{#1}\end{bfseries}\index{#1}}

\newcommand{\sse}{\mbox{$\:\subseteq\:$}}
\newcommand{\HB}{\hfill{$\Box$}}
\newcommand{\C}[1]{\mbox{$\{{#1}\}$}}

\newcommand{\set}[1]{{\{ #1 \}}}

\newcommand{\Formulas}{{\mathcal L}}

\newcommand{\Model}{\mathcal{M}}

\newcommand{\Agents}{Ag}

 \newcommand{\Situation}{\mathsf{s}}
 
 \newcommand{\init}{\mathfrak{i}}
 \newcommand{\Call}{\mathsf{c}}
 \newcommand{\Calld}{\mathsf{d}}

\newcommand{\CSequences}{{\bf C}}

 \newcommand{\CSequence}{{\mathbf c}}

 \newcommand{\CSequenced}{{\mathbf d}}

\newcommand{\pull}{{\triangleleft}}
\newcommand{\push}{\triangleright}
\newcommand{\pushpull}{{\scriptstyle\Diamond}}

\newcommand{\Atoms}{\mathsf{S}}
\newcommand{\Atomsb}{\mathsf{Q}}

\newcommand{\ag}[1]{{#1}}

\newtheorem{protocol}{Protocol}

\begin{document}

\mainmatter

\title{When Are Two Gossips the Same?}
\subtitle{Types of Communication in Epistemic Gossip Protocols\footnote{A shorter version of this paper appears in the Proceedings of the
22nd International Conference on Logic for Programming Artificial Intelligence and Reasoning (LPAR-22).}
} 
\author{Krzysztof R. Apt\inst{1,2} 
\and Davide Grossi\inst{3} 
\and Wiebe van der Hoek\inst{4} 
}
\institute{%
CWI, Amsterdam, The Netherlands
\and
MIMUW, University of Warsaw, Warsaw, Poland
\and
Bernoulli Institute, University of Groningen, Groningen, The Netherlands
\and
Department of Computer Science, University of Liverpool, Liverpool, U.K.
}

\maketitle


\begin{abstract} 
  We provide an in-depth study of the knowledge-theoretic aspects of
  communication in so-called gossip protocols. Pairs of agents
  communicate by means of calls in order to spread
  information---so-called secrets---within the group. Depending on the
  nature of such calls knowledge spreads in different ways within the
  group. Systematizing existing literature, we identify 18 different
  types of communication, and model their epistemic effects through
  corresponding indistinguishability relations. We then provide a
  classification of these relations and show its usefulness for an
  epistemic analysis in presence of different communication
  types. Finally, we explain how to formalise the assumption that the
  agents have common knowledge of a distributed epistemic gossip
  protocol.
\end{abstract}


\section{Introduction}


In the gossip problem \cite{tijdeman:1971,baker72gossips} a number of agents, each one
knowing a piece of information (a \emph{secret}) unknown to the
others, communicate by one-to-one interactions (e.g., telephone calls).
The result of each call is that the two agents involved in it learn
all secrets the other agent knows at the time of the call. The problem
consists in finding a sequence of calls which disseminates all the
secrets among the agents in the group.  It sparked a large literature
in the 70s and 80s
\cite{tijdeman:1971,baker72gossips,hajnal72cure,bumby81problem,seress86quick},
typically on establishing---in the above and other variants
of the problem---the minimum number of calls to achieve dissemination
of all the secrets. This number has been proven to be $2n - 4$, where
$n$, the number of agents, is at least $4$.

The gossip problem constitutes an excellent toy problem to study
information dissemination in distributed environments. A vast
literature on distributed protocols has taken up the problem and
analyzed it together with a wealth of variations including different
communication primitives (e.g., broadcasting instead of one-to-one
calls), as well as communication structures (networks), faulty
communication channels \cite{chlebus06robust}, and
probabilistic information transmission, where the spreading of gossips
is used to model the spread of an epidemic
\cite{bailey57mathematical,procaccia07gossip}.  Surveys are
\cite{fraigniaud94methods,hromkovic96dissemination,HHL88,hromkovic05dissemination}.


\paragraph{Background}

The present paper investigates a knowledge-based approach to the gossip
problem in a multi-agent system. Agents perform calls following
individual epistemic protocols they run in a distributed
fashion. These protocols tell the agents which calls to execute
depending on what they know, or do not know, about the information
state of the agents in the group. We call the resulting distributed
programs \emph{epistemic gossip protocols}, or {\em gossip protocols},
for short. Such protocols were introduced and studied in
\cite{ADGH14,apt15epistemic}.  `Distributed' means that each agent
acts autonomously, and `epistemic' means that the gossip protocols
refer to the agents' knowledge.  The reliance of these protocols on
epistemic properties makes them examples of so-called knowledge-based
protocols, as studied in the context of distributed systems
\cite{parikhetal:1985,kurki-suonio86towards,halpern92little,fagin97knowledge}.

Besides the aforementioned \cite{ADGH14,apt15epistemic}, a number of
papers have recently focused on epistemic gossip protocols. In
\cite{herzig17how} gossip protocols were studied that aim at achieving
higher-order shared knowledge, for example knowledge of level 2 which
stipulates that everybody knows that everybody knows all secrets. In
particular, a protocol is presented and proved correct that achieves
in $(k+1)(n-2)$ steps shared knowledge of level $k$.  Further, in
\cite{cooper16simple} gossip protocols were studied as an instance of
multi-agent epistemic planning that is subsequently translated into
the classical planning language PDDL. More recently,
\cite{ditmarsch17epistemic} presented a study of {\em dynamic} gossip
protocols in which the calls allow the agents not only to share the
secrets but also to share the communication channels (that is, who can
call whom).  In turn, \cite{apt17computational} studied the
computational complexity of distributed epistemic gossip protocols,
while \cite{AW18} showed that implementability, partial correctness,
termination, and fair termination of these protocols is decidable.

More broadly, the paper positions itself within the long-standing
tradition of analysis of distributed systems from the perspective
of epistemic logic \cite{fagin95reasoning,meyer95epistemic}. Such a
perspective has led in
\cite{parikhetal:1985,kurki-suonio86towards,fagin97knowledge} to a
useful level of abstraction allowing one to study a number of topics
in distributed computing from the knowledge theoretic perspective, in
particular protocols for the sequence transmission problem (for
instance the alternating bit protocol) in \cite{halpern92little},
coordination \cite{halpern90knowledge}, and secure communication
\cite{bataineh11abstraction}, to mention some. The characteristic
feature of these programs is that they use tests for knowledge.


\paragraph{Contributions}

The form of communication underpinning the epistemic gossip problem
may vary from work to work, and the above papers sometimes make
different assumptions on the nature of communication upon which the
considered protocols are based. Little attention has been devoted to a
systematic analysis, with the notable exception of
\cite{ditmarsch16parameters}, which singled out some of the key
informational assumptions on calls---specifically observability,
synchronicity and asynchronicity assumptions---and systematically
studied the effects of such assumptions on the aforementioned $2n-4$
call-length bound. 

It is our claim that research on epistemic gossip protocols can at
this point benefit from a systematisation of the key possible
assumptions that a modeler can make on the type of communication
(call) underpinning such protocols.  From an epistemic logic point of
view, each call type induces a specific notion of knowledge.  The
comparison of the resulting definitions of knowledge is of obvious
importance for the study of epistemic aspects of communication.

By `type of communication' we mean the way in which communication
takes place and may be observed, and to focus on it we disregard the
type of information exchanged (in particular, whether higher order
knowledge, or communication links may be exchanged---matters we do not
address), or the type of information the agents have initially at their
disposal (e.g., whether it is common knowledge what the number of
agents is).  

More specifically, here are the features we focus on.
First of all, a call between two agents takes place in the presence of
other agents. What these other agents become aware of after the call
is one natural parameter. We call it \emph{privacy}. The second
parameter, that we call \emph{direction}, clarifies in which direction
the information flows. Here we focus on three possibilities: they
exchange all information, one agent passes all information to the
other one, or one agent acquires all information available to the
other one. The final parameter of a call is what we call
\emph{observance}. It determines whether the agent(s) affected by the
call learn what information was held by the other agent prior to the
call. 

By a \emph{call type} we mean a combination of these three
parameters. What the agents know after a call, or more generally a
sequence of calls, depends on the assumed call type. This yields in
total 18 possibilities.  The paper provides a framework in which we
model these possibilities in a unified way.  This allows us to provide
in Theorem \ref{thm:classification} a complete classification of the
resulting indistinguishability relations.  This in turn makes it
possible to clarify in Propositions \ref{prop:1} and \ref{prop:newnew}
the effect of a call type on the truth of the considered
formulas. Additionally, we provide in Proposition \ref{prop:CK} a
natural proposal on how to incorporate into this framework an assumption
that the agents have common knowledge of the underlying protocol.

\paragraph{Paper outline}

Section \ref{sec:preliminaries} introduces gossip protocols by
example, and identifies the features of calls we will focus on.
Section \ref{sec:language} introduces the syntax and semantics of a
simple epistemic language to study communication and its effects in
gossip protocols, together with some motivating examples. Crucially
the semantics introduced is parametrised by the indistinguishability
relations which, for each call type, identify the call sequences that
the agents cannot distinguish.  These equivalence relations are
systematically introduced and defined in Section \ref{sec:relations},
and then compared in terms of their relative informativeness in
Section \ref{sec:comparison}.  The proposed systematisation is then
applied in Section \ref{sec:applications}: first, to deliver general
results on the analysis of how knowledge depends on the assumed call types
(Section \ref{sec:epistemic});
second, to offer a natural approach to the problem of modelling common
knowledge of protocols in the epistemic gossip setting (Section
\ref{sec:further}).  Finally, Section \ref{sec:conclusions} summarises
our results and charts several directions for future research.


\section{Knowledge-Based Gossip} \label{sec:preliminaries}

We start by recalling the notion of a gossip protocol, moving then to
introduce the formal set-up of the paper.

\subsection{Gossip protocols}

Gossip protocols aim at sharing knowledge between agents in a pre-described way. 
This is the paradigmatic setup:
\begin{quote}
{\it Six friends each know a secret. They can call each other by phone. In each call they exchange all the secrets they know. How many calls are needed for everyone to know all secrets?}
\end{quote}
Let us generalise this to the case of $n \geq 2$ agents and focus on
protocols that are {\em correct} (in the sense that they spread all
secrets). If $n = 2$, the two agents $\ag{a}$ and $\ag{b}$ need to
make only one phone call, which we denote by $\ag{ab}$ (`$\ag{a}$
calls $\ag{b}$'). For $n = 3$, the call sequence $ab, bc, ca$ will
do. Let us look at a protocol for $n \geq 4$ agents.

\begin{protocol} \label{protocol.nul} Choose four from the set of
  agents $\Agents$, say $a, b, c, d$, and one of those four, say
  $a$. First, $a$ makes $n - 4$ calls to each agent in
  $\Agents \setminus \{a,b,c,d\}$. Then, the calls $ab,cd,ac,bd$ are
  made. Finally $a$ makes another call to each agent from
  $\Agents\setminus\{a,b,c,d\}$. \label{prot.friends} \end{protocol}
This adds up to $(n-4) + 4 + (n-4) = 2n-4$ calls. For $n=6$ we get a
call sequence $ae,a\mathit{f},ab,cd,ac,bd,ae,a\mathit{f}$ of $8$
calls. All agents are then familiar with all secrets. It was shown
that less than $2n-4$ calls is insufficient to distribute all secrets
\cite{tijdeman:1971}.

The above protocol assumes that the agents can coordinate their
actions before making the calls.  But often such coordination is not
possible. Suppose some students of a given cohort receive an
unexpected invitation for a party. The members of the cohort may be
curious to find out about who received an invitation, in which case
they will resort to phone calls based on the knowledge, or better,
ignorance, they have about the secrets (in this context: extended
invitations) of others.  Since in such a distributed protocol several
agents may decide to initiate a call at the same time, we assume the
presence of an {\em arbiter} who breaks the ties in such cases. Let us
now consider such an epistemic protocol.

\begin{protocol}[Hear my secret]\label{protocol.mineen}
Any agent $a$ calls agent $b$ if $a$ does not know that $b$ is familiar with $a$'s secret.
\end{protocol}
This protocol has been proven in \cite{apt15epistemic} to terminate
and be correct, under specific assumptions on the type of
communication taking place during each call. In this paper we aim at
providing a systematic presentation of such assumptions and at an
analysis of their logical interdependencies.


\medskip

Throughout the paper we assume a fixed finite set $\Agents$ of at
least three \bfe{agents}. We further assume that each agent holds
exactly one \bfe{secret} and that the secrets are pairwise different.
We denote by $\mathsf{S}$ the set of all secrets, the secret of agent
$a$ by $A$, the secret of agent $b$ by $B$, and so on. A secret can be
any piece of data, for instance birthday, salary or social security
number.  Furthermore, we assume that each secret carries information
identifying the agent to whom this secret belongs.  So once agent $b$
learns secret $A$ she knows that this is the secret of agent $a$.

\subsection{Calls}  \label{subsec:calls}

In the context of gossip protocols calls constitute the sole form of
knowledge acquisition the agents have at their disposal.  Each
\bfe{call} concerns two agents, the \emph{caller} ($a$, below) and the
\emph{callee} ($b$, below). We call $a$ the \emph{partner} of $b$ in
the call, and vice versa. Any agent $c$ different from $a$ and $b$ is
called an {\em outsider}.  We study the following properties of calls:
\begin{itemize}
\item \bfe{privacy}, which is concerned with what the outsiders note about the call,
\item \bfe{direction}, which clarifies the direction of the information flow in the call,
\item \bfe{observance}, which clarifies, when an agent $a$ is informed by $b$, whether $a$ sees $b$'s secrets before 
adding them to her own set, or only sees the result of the fusion of the two sets of secrets.
\end{itemize}

More specifically, we distinguish three \bfe{privacy degrees} of a call where agent $a$ calls $b$:

\begin{itemize}

\item $\myprivacy{1}$: every agent $c \neq a,b$ notes that $a$ calls $b$,

\item $\myprivacy{2}$: every agent $c \neq a,b$ notes that some call takes place,
though not between whom,


\item \myprivacy{3}: no agent $c \neq a,b$ notes that a call is taking place.

\end{itemize}

Intuitively, these degrees can be ordered as
$\myprivacy{1} <_\prtype \myprivacy{2} <_\prtype \myprivacy{3}$, with
$\myprivacy{1}$ meaning no privacy at all, $\myprivacy{2}$ ensuring
anonymity of the caller and callee, and $\myprivacy{3}$ denoting full
privacy.  Conversely, from the perspective of the agents not involved
in the call, a call with the privacy level $\myprivacy{1}$ is the most
informative, while a call with the privacy level $\myprivacy{3}$ is
the most opaque. 

We distinguish three \bfe{direction types}, in short
\bfe{directions}, of a call:

\begin{itemize}

\item {\bfe{push}}, written as $\push$.  As a result of the call the callee learns all the secrets held by the
  caller. 

\item {\bfe{pull}}, written as $\pull$. As a result of the call the caller learns all the secrets held by the callee. 
  
\item {\bfe{push-pull}}, written as $\pushpull$. As a result of the call the caller and the callee learn each other's
  secrets. 
\end{itemize}

Depending on the direction of a call between $a$ and $b$, one or both
agents can learn \emph{directly} new information thanks to it. We say
that these are the agents {\em affected} in the call. 
More formally, 
an agent $a$ is \bfe{affected} by a call $\Call$ if $\Call$ is
one of the following forms:
\[
a \pushpull b, b \pushpull a, b \push a, \mathit{or}~a\, \pull \, b.
\]

Intuitively, $a$ is affected by the call if it can affect the set
of secrets $a$ is familiar with.
This brings us to two possible levels of \bfe{observance} of a call:

\begin{itemize}

\item \observance{after}: During the call the affected agent(s) add
  the secrets of their partner to their own secrets, and only
  \emph{after} that, inspect the result.\footnote{This mode is akin to
    the caller and callee interacting through a third party, who first
    collects the caller's and callee's secrets separately, and then
    shares their union with the affected agent(s).  We are indebted to
    R.~Ramezanian for this observation.}

\item \observance{before}: During the call the affected agent(s)
  inspect the secrets of their partner \emph{before} adding them to
  their own secrets.
\end{itemize}

Intuitively, the observance level $\observance{after}$ is less
informative for an affected agent than $\observance{before}$, because in
the latter case she also learns which secrets were known to the
other agent before adding them to the secrets she is familiar with.
Let
\begin{itemize}
\item $\prset = \{\myprivacy{1}, \myprivacy{2}, \myprivacy{3}\}$, 

\item $\dirset = \{\pushpull,\triangleleft,\triangleright\}$,

\item $\oset = \{\observance{after},\observance{before}\}$.

\end{itemize}
Each call between agents $a$ and $b$ is of the shape
$\callupone{ab}{\tau}$, where
$\tau = (\prtype, \dirtype, \otype) \in \prset \times \dirset \times
\oset$ is called its \bfe{type}.  So we defined in total 18 call
types.  To clarify their effect on communication we will elaborate on
some representative call types in
Examples~\ref{exa:preliminary1}--\ref{exa:preliminary3}.

The types $(\myprivacy{1}, \pushpull, \observance{before})$ and
$(\myprivacy{2}, \pushpull, \observance{before})$ were studied in
\cite{ADGH14} while the types
$(\myprivacy{3}, \pushpull, \observance{after})$,
$(\myprivacy{3}, \push, \observance{after})$, and
$(\myprivacy{3}, \pull, \observance{after})$, were analyzed in
\cite{apt15epistemic}. 
For a type $\tau$ like
$(\myprivacy{1},\pushpull,\observance{before})$, we define
$\tau(\prtype) = \myprivacy{1}, \tau(\dirtype) = \pushpull$ and
$\tau(\otype) = \observance{before}$.






Often, the call type (or parts of it) is (are) clear from the context,
and we omit it (them).  In our examples, at the level of calls, we
often only explicitly mention the direction type. Given a
call between $a$ and $b$ we shall sometimes write it simply as $ab$
for the direction type $\pushpull$, $a \push b$ for the
direction type $\push$ and $a\, \pull\, b$ for the direction type $\pull$.





\section{Language and Semantics} \label{sec:language}

In this section we introduce a modal language for epistemic gossip and its formal semantics.

\subsection{Modal language}

We are interested in determining agents' knowledge after a sequence of
calls took place. To this end we use the following modal language
$\lang$ for epistemic logic:
\[ \begin{array}{lcl}
\phi & ::= & F_a S  \mid \neg \phi \mid \phi \wedge \phi \mid \phi \vee \phi \mid K_a \phi,
   \end{array} \] 
 where $a \in \Agents$ and  $S \in \Atoms$. 
  
 In what follows we refer to the elements $\phi$ of $\lang$ as
 \bfe{epistemic formulas}, or in short, just formulas.  We read
 $F_a S$ as `agent $a$ is familiar with the secret $S$' (or `$S$
 belongs to the set of secrets $a$ has learned') and $K_a \phi$ as
 `agent $a$ knows that formula $\phi$ is true'.  So $\lang$ is an
 epistemic language with the atomic formulas of the form $F_a S$.

 The above language was introduced in \cite{apt15epistemic}. It is a
 modification of the language introduced in \cite{ADGH14}.





\begin{example} \label{exa:expert}
Consider the statement that agent $a$ is familiar with all the secrets.
This can be expressed as the formula
\[
\bigwedge_{b \in \Agents} F_a B
\]
that we subsequently abbreviate to $Exp_a$ (``$a$ is an expert'').

Here and elsewhere for simplicity we refer in the conjunction limits only
to agents and not to their secrets. This convention allows us to write more complex
statements, for instance that each agent is familiar only with her own secret.
This can be expressed as the formula
\begin{equation}
\bigwedge_{a \in \Agents}( F_a A \land \bigwedge_{b \in \Agents, b \neq a} \neg F_a B). \label{eq:init}
\end{equation}

Finally, consider the statement that for each agent, say $a$, it is
not the case that $a$ is an expert and each other agent is familiar
with at most her own secret and that of $a$.  This can be expressed as
the formula
\[
\bigwedge_{a \in \Agents} \neg (Exp_a \land \bigwedge_{b,c \in \Agents, \: |\{a, b,c\}| = 3}\neg F_b C).
\]
\HB
\end{example}

\smallskip

Next, we clarify the use of the knowledge
operators. In the presented reasoning we assume that the agents have
the knowledge of the underlying call type. In all cases we assume that
the initial situation is the one in which every agent is only familiar
with her own secret, that is, we assume \eqref{eq:init} to be true
for each agent before any communication takes place.  The examples
provide intuitions about how agents' knowledge is influenced by the
types of calls underpinning their communications. Such intuitions will
then be formalised in Section \ref{sec:semantics}.

\begin{example} \label{exa:new}
  Initially, each agent is familiar with her secret and each agent
  knows this fact. Additionally, she does not know that any other
  agent is familiar with a secret different from her own.  This can be
  expressed by means of the formula
\[
\bigwedge_{a \in \Agents} (\bigwedge_{b \in \Agents} 
K_a F_b B \land \bigwedge_{b,c \in \Agents, a \neq b, b \neq c} \neg K_a F_b C) 
\]
that holds initially, for all call types.
\HB
\end{example}

\begin{example} \label{exa:preliminary1}
Suppose there are four agents, $a, b, c$ and $d$.
Consider the call type is 
$(\myprivacy{1},\pushpull,\observance{after})$.
Assume the call sequence $ab, bc$. 

Let us reason from the perspective of agent $d$. Because of the
assumed privacy level, after the first call, $ab$, agent $d$ knows that
both agents $a$ and $b$ are familiar with $A$ and $B$. This can be
expressed as the formula
\[
K_d (F_a A \land F_a B \land F_b A \land F_b B).
\]
This then implies that after the second call, $bc$, agent $d$ also 
knows that both $b$ and $c$ are familiar with
$A, B$, and $C$. Agent's $d$ factual knowledge after the second call
can be expressed as the formula
$K_d \phi$, where 
\[
\phi = F_a A \land F_a B \land \bigwedge_{i \in \{b, c\}, \: j \in \{a, b, c\}} F_i J.
\]

In fact, because of the assumed privacy level $\myprivacy{1}$,
how the knowledge evolves during communication is completely
transparent to all agents. Hence after both calls everybody knows $\phi$, i.e., 
\[
\bigwedge_{a \in \Agents} K_a \phi.
\]
An analogous argument applies for the call type $(\myprivacy{1},\pushpull,\observance{before})$.

Suppose now that the privacy level is $\myprivacy{2}$.  Then we
cannot conclude the formula $K_d \phi$ after the second call, since agent $d$ only knows
then that two calls took place, but not between which pairs of
agents. In fact, in this case we can only conclude (note that the same call can be made twice):
\[
K_d (\bigvee_{a, b \in \Agents\setminus\{d\}, a \neq b} (F_a B \land F_b A)).
\]

Finally, if the privacy level is $\myprivacy{3}$, then $d$ is not
aware of the calls $ab$ and $bc$. She considers it possible that
$a,b,c$ are already familiar with all secrets except her own, but also
considers it possible that all other agents only know their own
secret. As she has not yet been involved in any call, she knows that
they are not familiar with $D$.  

So after the call sequence $ab, bc$ agent's $d$ knowledge can be expressed as
\[
K_d (\bigwedge_{e \in \Agents\setminus\{d\}} (F_e E \land \neg F_e D)).
\]
\HB
\end{example}

\begin{example} \label{exa:preliminary2}
Suppose there are three agents, $a, b$ and $c$.  Consider the two
call types $(\myprivacy{3},\pushpull,\otype)$, where $\otype \in \oset$,
and assume the call sequence
$ac, bc, ab$.
After it the agents $a$ and $b$ (and $c$ too) are 
familiar with all the secrets, which can be expressed as the
formula 
\[
\phi = Exp_a \land Exp_b,
\]
and both know this fact, which
can be expressed as $K_a \phi \land K_b \phi$. 

If the observance of the calls is \observance{before}, agent $a$ also
learns that prior to the call $ab$ agent $b$ was familiar with $a$'s
secret, i.e., with $A$. This allows $a$ to conclude that agent $b$ was
involved in a call with $c$ and hence agent $c$ is familiar with $B$.
We can express this as 
\[
K_a F_c B.
\]

Contrast the above with the situation when the observance is
$\observance{after}$.  Although again after the considered call
sequence both agents $a$ and $b$ are familiar with all the secrets,
now agent $a$ cannot conclude that agents $b$ and $c$
communicated. Hence agent $a$ does not know whether agent $c$ is
familiar with $B$, i.e., the formula $K_a F_c B$ is not true.

In both cases agent $c$ (who also is an expert) does not know that agents $a$ and $b$
communicated, so she does not know that they are experts. In other
words, the formula $K_c \phi$ is not true. This changes when the
privacy degree is \myprivacy{1}, i.e., in that case the formula $K_c \phi$ is
true. Moreover, because there are three agents, the same conclusion
holds when the privacy degree is \myprivacy{2}. However, the last
conclusion does not hold anymore when there are more than three
agents.
\HB
\end{example}



\begin{example} \label{exa:preliminary3}
Assume the same call sequence as in the previous example but suppose that the
call parameters are now $(\myprivacy{1}, \pull, \otype)$, where $\otype \in \oset$.
So we consider now the call sequence
$\cc = a \, \pull\, c, b \, \pull\, c, a \, \pull\, b$.

Because of the assumed privacy level, after this call sequence
agent $a$ knows that agent $b$ learned the secret $C$ and
agent $c$ knows that agent $a$ learned the secret $B$, i.e., the following holds after $\cc$
\[
K_a F_b C \land K_c F_a B.
\]

Suppose now the privacy degree is \myprivacy{2} and the observance is
$\beta$.  Then we only have $K_a F_b C$ as agent $a$ cannot
distinguish $\cc$ from
$a \, \pull\, c, c \, \pull\, b, a \, \pull\, b$. Clearly, $K_c F_a B$
does not hold after $\cc$ as agent $c$ cannot distinguish $\cc$ from
$a \, \pull\, c, c \, \pull\, b, b \, \pull\, a$.

Finally, if the privacy degree is \myprivacy{3} then
for the same reason $K_c F_a B$ does not hold after $\cc$ either.
%
\HB
\end{example}

We conclude that what the agents know after a call sequence crucially
depends on the parameters of the calls.  Further, the precise effect of a
single call on the agents' knowledge is very subtle, both for the
agents involved in it and for the outsiders.


\subsection{Semantics}
\label{sec:semantics}

We provide now a formal semantics for the modal language $\lang$.  
It is parameterized by a call type $\tau$.

\paragraph{Gossip situations and calls}

First we recall the following crucial notions introduced in
\cite{apt15epistemic}.  A \bfe{gossip situation} is a sequence
$\Situation = (\Atomsb_a)_{a \in \Agents}$, where
$\Atomsb_a \sse \Atoms$ for each agent $a$.  Intuitively, $\Atomsb_a$
is the set of secrets agent $a$ is familiar with in the situation
$\Situation$.  Given a gossip situation
$\Situation = (\Atomsb_a)_{a \in \Agents}$, we denote $\Atomsb_a$ by
$\Situation_a$.  The \bfe{initial gossip situation} is the one in
which each $\Atomsb_a$ equals ${\{A\}}$ and is denoted by $\init$ (for
``initial'').  The initial gossip situation reflects the fact that
initially each agent is familiar only with her own
secret. 



Each call transforms the current gossip situation by possibly
modifying the set of secrets the agents involved in the call are
familiar with. The definition depends solely on the direction of the
call.

\begin{definition} \label{def:effects} 

  The application of a call $\Call$ to a gossip situation $\Situation$
  is defined as follows, where
  $\Situation := (\Atomsb_a)_{a \in \Agents}$:
\begin{description}

\item[\fbox{$\Call = ab$}] $\Call(\Situation) = (Q'_a)_{a \in \Agents}$, where 
$\Atomsb'_a = \Atomsb'_b = \Atomsb_a \cup \Atomsb_b$,
$\Atomsb'_c = \Atomsb_c$, for $c \neq a,b$.

\item[\fbox{$\Call = a \push b$}] $ \Call(\Situation) =(Q'_a)_{a \in \Agents}$, where
$\Atomsb'_b = \Atomsb_a \cup \Atomsb_b$,
$\Atomsb'_a = \Atomsb_a$,
$\Atomsb'_c = \Atomsb_c$, for $c \neq a,b$.

\item[\fbox{$\Call = a\, \pull \, b$}] $\Call(\Situation) = (Q'_a)_{a \in \Agents}$, where
$\Atomsb'_a = \Atomsb_a \cup \Atomsb_b$,
$\Atomsb'_b = \Atomsb_b$,
$\Atomsb'_c = \Atomsb_c$, for $c \neq a,b$.
\end{description}

\end{definition}

This definition captures the meaning of the direction type: for $ab$
the secrets are shared between the caller and callee , for
$a \push b$ they are pushed from the caller to the callee, and for
$a\, \pull\, b$ they are retrieved by the caller from the callee.  Note
that $(a \pushpull b)(\Situation) = (b \pushpull a)(\Situation)$
and $(a \push b)(\Situation) = (b\, \pull\, a)(\Situation)$, as
expected.

In turn, the privacy degree of a call captures what outsiders of the
call learn from it and the observance level determines
informally what caller and callee can learn about each other's calling
history. The meaning of these two parameters will be determined by
means of the appropriate equivalence relations between call
sequences. 


A \bfe{call sequence} is a {\em finite} sequence of calls, all of the same
call type.  The empty sequence is denoted by $\epsilon$.  We use
$\CSequence$ to denote a call sequence and $\CSequences^{\calltype}$
to denote the set of all call sequences of call type $\calltype$.  
Given the call sequence $\CSequence$ and a call $\Call$, $\CSequence.\Call$ denotes the sequence obtained by appending $\CSequence$ with $\Call$.



The result of applying a call sequence $\CSequence$ to a situation
$\Situation$ is defined by induction using Definition
\ref{def:effects}, as follows
\begin{description}
\item{[Base]} $\epsilon(\Situation) := \Situation$,

\item{[Step]} $\CSequence.\Call(\Situation) := \Call(\CSequence(\Situation))$.  
\end{description}
Note that this definition does not depend on the privacy degree and observance of the calls.

\begin{example}
Let  $\Agents$ be $\set{a,b,c}$.  We use the following
concise notation for gossip situations. Sets of secrets will be
written down as lists. E.g., the set $\set{A, B, C}$ will be written
as $ABC$. Gossip situations will be written down as lists of lists of
secrets separated by dots. E.g., $\init = A.B.C$ and the gossip
situation $(\set{A,B}, \set{A,B}, \set{C}$) will be written as
$AB.AB.C$. So, $(ab)(A.B.C) = AB.AB.C$,
$(ab, ca)(A.B.C) = ABC.AB.ABC$ and $(ab, ca, ab)(A.B.C) = ABC.ABC.ABC$.

%
\HB
\end{example}





\paragraph{Truth of formulas}


We illustrated in Examples
\ref{exa:preliminary1}--\ref{exa:preliminary3} that each call has an
effect on the knowledge of the agents. After a sequence of calls took
place the agents may be uncertain about the current gossip situation
because they do not know which call sequence actually took place. This
leads to appropriate indistinguishability relations that allow us to
reason about the knowledge of the agents. This
is in a nutshell the basis of the approach to epistemic gossip protocols
put forth in \cite{apt15epistemic}, and upon which we build here.

To clarify matters consider the situation analyzed in Example
\ref{exa:preliminary2}. We noticed there that depending on the assumed
observance level the knowledge of agent $a$ differs. This has to do
with the call sequences the agent considers possible.  If the call type is
$(\myprivacy{3},\pushpull,\alpha)$ agent $a$ cannot distinguish
between the call sequences $ac, ab$ and $ac, bc, ab$.  Indeed, after
both sequences she is familiar with all the secrets but she cannot
determine whether agents $b$ and $c$ communicated. From her
perspective both call sequences are possible, that is, she
cannot distinguish between them. In contrast, if the call type is
$(\myprivacy{3},\pushpull,\beta)$ agent $a$ can distinguish between
these two call sequences, which has in turn an effect on her knowledge.

In general, to determine what agents know after a call sequence we
need then to consider an appropriate equivalence relation between the
call sequences.  Let $\CSequence$ and $\CSequenced$ be two call
sequences of call type $\calltype$ and $a$ an agent. The statement
$\CSequence \sim^\calltype_a \CSequenced$ 
informally says that agent
$a$ cannot distinguish between $\CSequence$ and $\CSequenced$.  The
definition of $\sim^\calltype_a$ crucially depends on the call type
$\calltype$ and is provided in the next subsection. Here we assume
that it is given and proceed to define the truth of the formulas of
the language $\lang$ with respect to a \bfe{gossip model} (for a given
set of agents $\Agents$)
$\Model^\calltype = (\CSequences^\calltype, \set{\sim^\calltype_a}_{a
  \in \Agents})$ and a call sequence $\CSequence$ as follows:


\begin{definition}\label{def:truth}
  Let $\Model^\calltype$ be a gossip model for a call type $\calltype$
  and a set of agents $\Agents$, and let
  $\CSequence \in \CSequences^{\calltype}$.  The truth relation for
  language $\lang$ is inductively defined as follows (with Boolean
  connectives omitted):
\begin{eqnarray*}
  (\Model^\calltype, \CSequence) \models  F_a S & \mbox{iff} & S \in \CSequence(\init)_a, \\
  (\Model^\calltype, \CSequence) \models  K_a \phi &  \mbox{iff}  & \forall \CSequenced \in \CSequences^{\calltype} \mbox{ such that } \CSequence \sim^{\calltype}_a \CSequenced, ~ (\Model^\calltype, \CSequenced) \models \phi.
\end{eqnarray*}
Since the gossip model is clear from the context, we will from now on
write $\CSequence \models^\calltype \phi$ for
$(\Model^\calltype, \CSequence) \models \phi$.  We also write
$\Model^\calltype \models \phi$ ($\phi$ is valid in
$\Model^\calltype$) if for all
$\CSequence \in \CSequences^{\calltype}$ we have
$\Model^\calltype, \CSequence \models \phi$.
\end{definition}

So the formula $F_a S$ is true after a sequence of calls $\CSequence$
whenever agent $a$ is familiar with the secret $S$ in the gossip
situation generated by $\CSequence$ applied to the initial gossip
situation $\init$. The knowledge operator $K_a$ is interpreted as
is customary in the multimodal $S5_n$ logic (see, e.g.,
\cite{meyer95epistemic,ditmarsch07dynamic}), so using the equivalence relation
$\sim^{\calltype}_a$.

It is important to notice that to determine the truth of a propositional formula
(so in particular to determine which secrets an agent is familiar with) only 
the direction parameter of the type of the calls is used. In contrast, to 
determine the truth of formulas involving the knowledge operator all three parameters
of the call type are needed, through the definition of the $\sim^{\calltype}_a$
relations, to which we turn next.


\section{Indistinguishability of Call Sequences} \label{sec:relations}

Below we say that an agent $a$ is \bfe{involved} in a call $\Call$,
and write $a \in \Call$, if $a$ is one of the two agents involved in
it, i.e., if it is either a caller or a callee in $\Call$.  So agent
$a$ is involved but not affected (a notion introduced in Section
\ref{sec:preliminaries}) by a call $\Call$ if $\Call = a \push b$ or
$\Call = b\, \pull \, a$  for some agent $b$.


\subsection{The $\sim^{\calltype}_a$ relations}

For every call type $\calltype$ and agent $a$ we define the
indistinguishability relation
$\sim^\calltype_a \subseteq \CSequences^{\calltype} \times
\CSequences^{\calltype}$ in two steps.  First we define the auxiliary
relation $\approx^\calltype_a$ (Definition
\ref{def:aux}). Intuitively, the expression
$\CSequence \approx^\calltype_a \CSequenced$ can be interpreted as
``from the point of view of $a$, if $\CSequence$ is an (epistemically)
possible call sequence, so is $\CSequenced$, and vice versa''.  Then, we define
$\sim^\calltype_a$ as the least equivalence relation that contains
$\approx^\calltype_a$.

\begin{definition} \label{def:aux}
Let $a \in \Agents$ and fix a type $\calltype$. The relation $\approx^\calltype_a$ is the smallest subset of $\CSequences^{\calltype} \times
\CSequences^{\calltype}$ satisfying the following conditions:
\begin{description}
\item{[Base]} $\epsilon \approx^{\calltype}_a \epsilon$.

\item{[Step]} Suppose that $\CSequence \approx^\calltype_a \CSequenced$ and let $\Call$ and $\Calld$ be calls. 
\[
\begin{array}{ll}
\text{[Step-out}^\calltype\text{\!]}  & \mbox{if } \outcond{\calltype}{a}{out}\mbox{ then }\myconclout{\calltype}{a}, \\
\text{[Step-in}^\calltype\text{\!]}  & \mbox{if } \mycond{\calltype}{a}\mbox{ then }\myconclin{\calltype}{a}, \\
\end{array}
\]
\end{description}
where the used relations are defined in Table \ref{table:een}.
($b$ is there the partner of $a$ in the call $\Call$.)
\end{definition}
\begin{table}[hbt]
\begin{center}
Agent $a$ is not involved in the last call:
\vspace*{0.1cm}

 \begin{tabular}{|| c | c | c  ||} 
 \hline
 $\calltype(\prtype)$ & $\outcond{\calltype}{a}{out}$        & $\myconclout{\calltype}{a}$  \\ [0.7ex] 
 \hline\hline
 $\myprivacy{1}$      &$ a \not\in \Call$                    & $\CSequence.\Call \approx^\calltype_a \CSequenced .\Call$ \\ [0.5ex]
 \hline
$\myprivacy{2}$       &$ a \not\in \Call, a \not\in \Calld$  & $\CSequence.\Call \approx^\calltype_a \CSequenced .\Calld$ \\ [0.5ex]\hline
$\myprivacy{3}$       &$ a \not\in \Call$                    & $\CSequence.\Call \approx^\calltype_a \CSequenced$, $\CSequence \approx^\calltype_a \CSequenced.\Call$ \\ [0.5ex]\hline
 \end{tabular}
%

\bigskip

Agent $a$ is involved in but not affected by the last call:
\vspace*{0.1cm}

 \begin{tabular}{|| c | c ||} 
 \hline
 $\mycond{\calltype}{a}$ & $\myconclin{\calltype}{a}$  \\ [0.5ex] 
 \hline\hline
 $\Call \in \{a \push b, b\, \pull \, a\}$  & $\CSequence.\Call \approx^{\tau}_a \CSequenced.\Call$ \\ [0.5ex]\hline
 \end{tabular}
%

\bigskip

Agent $a$ is involved in and affected by the last call:
\vspace*{0.1cm}

 \begin{tabular}{||c | c | c ||} 
 \hline
  $\calltype(\otype)$ & $\mycond{\calltype}{a}$ & $\myconclin{\calltype}{a}$  \\ [0.5ex] 
 \hline\hline
 $\observance{after}$ & $\Call \in \{a \pushpull b, b \pushpull a, b \push a, a\, \pull \, b\}$,
 & $\CSequence.\Call \approx^{\tau}_a \CSequenced.\Call$ \\
 &  $\CSequence.\Call(\init)_a = \CSequenced.\Call(\init)_a$ &  \\[0.5ex]\hline
 $ \observance{before}$ & $\Call \in \{a \pushpull b, b \pushpull a, b \push a, a\, \pull \, b\}$,
 & $\CSequence.\Call \approx^{\tau}_a \CSequenced.\Call$ \\
 & $\CSequence(\init)_b = \CSequenced(\init)_b$ & \\ [0.5ex]\hline
 \end{tabular}
%
\end{center}

\caption{Defining indistinguishability of call sequences}\label{table:een}\label{table:twee}\label{table:drie}
\end{table}






The definition of $\approx^\calltype_a$ captures the complex effect of
each of the three parameters of a call type on the knowledge of an
agent. Let us discuss it now in detail.

The Base condition is clear. 
Consider now the $\text{Step-out}^\calltype$ clause which refers to
Table \ref{table:een}, top.  Suppose that
$\CSequence \approx^\calltype_a \CSequenced$.  Consider first the
privacy type $\myprivacy{1}$. According to its informal description the
condition $a \not\in \Call$ means that agent $a$ is not involved in
the call $\Call$ but knows who calls whom.  The conclusion
$\CSequence.\Call \approx^\calltype_a \CSequenced .\Call$ then
coincides with this intuition.

Consider now the privacy type $\myprivacy{2}$.  The conditions
$a \not\in \Call$ and $a \not\in \Calld$ mean that agent $a$ is not
involved in the calls $\Call$ and $\Calld$, thus according to the
informal description of $\myprivacy{2}$ she cannot distinguish between
these two calls.  This explains the conclusion
$\CSequence.\Call \approx^\calltype_a \CSequenced .\Calld$.  Note
that this conclusion is not justified for the privacy type
$\myprivacy{1}$ because if $\Call \neq \Calld$ then agent $a$ can
distinguish between these two calls, so a fortiori between the call
sequences $\CSequence.\Call$ and $\CSequenced .\Calld$.

Finally, consider the privacy type $\myprivacy{3}$.  According to its
informal description, the condition $a \not\in \Call$ means that agent
$a$ is not aware of the call $\Call$.  This justifies the conclusions
$\CSequence.\Call \approx^\calltype_a \CSequenced$ and
$\CSequence \approx^\calltype_a \CSequenced.\Call$.

Next, consider the $\text{Step-in}^\calltype$ clause.  It spells the
conditions that allow one to extend the $\approx^\calltype_a$
relation in case agent $a$ is involved in the last call, $\Call$.
Table \ref{table:twee}, middle, formalises the intuition that when
agent $a$ is not affected by the call $\Call$, then we can conclude
that $\CSequence.\Call \approx^\calltype_a \CSequenced .\Call$.

Table \ref{table:drie}, bottom, focuses on the remaining case.  Consider first
the observance $\alpha$. According to its informal description,
affected agents incorporate the secrets of their partner with their own
secrets and then inspect the result.  So we check what secrets agent
$a$ is familiar with after the call sequences $\CSequence$ and
$\CSequenced$ are both extended by $\Call$.  If these sets are equal,
then we can conclude
that $\CSequence.\Call \approx^\calltype_a \CSequenced .\Call$.

In the case the observance is $\beta$, the informal description
stipulates that the agent inspects the set of secrets of the call
partner before incorporating them with their own secrets.  So we
compare these sets of secrets after, respectively, the call sequences
$\CSequence$ and $\CSequenced$ took place.  If these sets are equal,
then we conclude that
$\CSequence.\Call \approx^\calltype_a \CSequenced .\Call$.  This
explains why in this case a reference to agent $b$ is made in
$\mycond{\calltype}{a}$.


\subsection{Examples and a useful observation}

\begin{example} \label{exa:reconsidered}
  We first illustrate Table \ref{table:een}, top, by analyzing
  situations in which the considered agent is not involved in the last
  call.  Assume four agents, $a, b, c$ and $d$. 
  
  Suppose that the privacy of $\tau$ is $\myprivacy{1}$. We have
  $ab, bc \not\sim^{\tau}_a ab, cd$, because
  $ab, bc \not\approx^{\tau}_a ab, cd$ as $bc \neq cd$ and
  $bc \neq dc$. So we fail to apply Table \ref{table:een}, top, first
  row and the transitive reflexive closure does not give us that either.

  On the other hand, if the privacy of $\tau$ is $\myprivacy{2}$, we
  have $ab, bc \sim^{\tau}_a ab, cd$, because
  $ab, bc \approx^{\tau}_a ab, cd$, as $a \not\in bc$, $a \not\in cd$
  and $ab \sim^{\tau}_a ab$ (Table \ref{table:een}, top, second row).
  On the other hand, $ab, bc \not\sim^{\tau}_a ab, cd, bc$ as now the
  clause in the second row fails to apply, as the lengths of the
  compared sequences are different.

  Finally, if the privacy of $\tau$ is $\myprivacy{3}$, we of course
  also have $ab, bc \sim^{\tau}_a ab, cd$ for the same reason as in
  the previous paragraph, but we now also have
  $ab, bc \sim^{\tau}_a ab, cd, bc$, because
  $ab, bc \approx^{\tau}_a ab, cd, bc$. Indeed, we have
  $ab \sim^{\tau}_a ab$ and hence by Table \ref{table:een}, top, third
  row, applied three times, first $ab \sim^{\tau}_a ab, cd$, then
  $ab, bc \sim^{\tau}_a ab, cd$, and finally
  $ab, bc \sim^{\tau}_a ab, cd, bc$.  \HB
\end{example}

\begin{example} \label{exa:fiets} To illustrate Table \ref{table:een},
  middle, consider the same four agents and sequence
  $d \push c, b \push c$, and $\myprivacy{2}$. Then
  $d \push c, b \push c \sim^{\tau}_b c \push d, b \push c$, because
  agent $b$ is involved in the second call but not affected (Table
  \ref{table:een}, middle), and $d \push c \sim^{\tau}_b c \push d$,
  because $b \not \in d \push c$ and $b \not \in c \push d$ (Table
  \ref{table:een}, top, second row).  \HB
\end{example}

\begin{example} \label{exa:zwem} Now consider Table \ref{table:een},
  bottom. The difference between observances $\alpha$ and $\beta$ is
  seen in Example \ref{exa:preliminary3}. For the observancy $\alpha$
  we have that $a \, \pull\, c, b \, \pull\, c, a \, \pull\, b$
  $\sim^\tau_a$ $a \, \pull\, c, c \, \pull\, b, a \, \pull\, b$,
  because agent $a$ is afterwards familiar with the same
  secrets on the lefthand side and the righthand side, namely
  $A,B,C$ (Table \ref{table:een}, bottom, first row). On the other
  hand, for the observancy $\beta$ we get
  $a \, \pull\, c, b \, \pull\, c, a \, \pull\, b \not\sim^\tau_a a \,
  \pull\, c, c \, \pull\, b, a \, \pull\, b$, because
  $a \, \pull\, c,$
  $b \, \pull\, c \not\sim^\tau_b a \, \pull\, c, c \, \pull\, b$
  (note that this concerns indistinguishability for agent $b$, not
  $a$); here the second row of Table \ref{table:een}, bottom, applies.
  
  As a final example,  we have that
  $d \push c, b \push c \not\sim^{\tau}_c c \push d, b \push c$, because
 $c$ is involved but not affected in call $c \push d$, while it is involved and affected in call $d \push c$.
 Observe that after $d \push c, b \push c$ agent $c$ is
  familiar with the secrets $B,C,D$, whereas after $c \push d, b \push c$
  agent $c$ is only familar with $B,C$. 
\HB
\end{example}

Let us focus now on some properties of the $\sim^{\tau}_a$ equivalence relations.

\begin{note} \label{not:1}
 For all agents $a$ and call types $\tau$
\[
\sim^{\tau}_a = (\approx^{\tau}_a)^*,  
\]
where $^*$ is the transitive, reflexive closure operation on binary
relations.  
\end{note}
\begin{proof}
A straightforward proof by induction show 
that each $\approx^{\tau}_a$ relation is symmetric. This implies the claim.
\HB
\end{proof}

The following observation will be needed later.

\begin{proposition} \label{prop:a}
For all call types $\tau$ if
$\CSequence \sim^{\tau}_a \CSequenced$, then
$\CSequence(\init)_a =  \CSequenced(\init)_a$.
\end{proposition}

\begin{proof}
By Note \ref{not:1} it is sufficient to prove the conclusion
under the assumption that $\CSequence \approx^{\tau}_a \CSequenced$.

We proceed by induction on the sum $k$ of the lengths $|\CSequence| + |\CSequenced|$
of both sequences. If $k = 0$, then
$\CSequence = \CSequenced = \epsilon$, so the claim holds.  Suppose
the claim holds for all pairs of sequences such that the sum of their
lengths is $< k$ and that $k > 0$, $|\CSequence| + |\CSequenced| = k$ and
$\CSequence \approx^{\tau}_a \CSequenced$.  By definition
$\approx^{\tau}_a$ is the smallest relation satisfying the Base and
Step conditions of Definition \ref{def:aux}.  Let $\Call$ be the last
call of $\CSequence$ or of $\CSequenced$ if $\CSequence$ is empty.

If agent $a$ is not involved in $\Call$, then four cases arise,
depending on the form of $\CSequence$ and $\CSequenced$. We consider
one representative case, when $\CSequence$ is of the form
$\CSequence'.\Call$, where $\CSequence' \approx^{\tau}_a \CSequenced$.
Then by the assumption about $\Call$ and the induction hypothesis
\[
\CSequence(\init)_a =
\CSequence'.\Call(\init)_a = \CSequence'(\init)_a = \CSequenced(\init)_a. 
\]

If agent $a$ is involved in but not affected by the last call, then
$\CSequence$ is of the form $\CSequence'.\Call$, $\CSequenced$ is of
the form $\CSequenced'.\Call$,
$\Call \in \{a \push b, b\, \pull \, a\}$ and
$\CSequence' \approx^{\tau}_a \CSequenced'$.  Then by the form of
$\Call$ and the induction hypothesis
\[
\CSequence(\init)_a =
\CSequence'.\Call(\init)_a = \CSequence'(\init)_a = \CSequenced'(\init)_a = 
\CSequenced'.\Call(\init)_a  = \CSequenced(\init)_a.
\]

Finally, if agent $a$ is involved in and affected by the last call,
then $\CSequence$ is of the form $\CSequence'.\Call$, $\CSequenced$ is
of the form $\CSequenced'.\Call$,
$\Call \in \{a \pushpull b, b \pushpull a, b \push a, a\, \pull \,
b\}$ and $\CSequence' \approx^{\tau}_a \CSequenced'$.

If $\tau(\otype) = \alpha$, then by assumption
$\CSequence'.\Call(\init)_a = \CSequenced'.\Call(\init)_a$, i.e.,
$\CSequence(\init)_a = \CSequenced(\init)_a$.  If
$\tau(\otype) = \beta$, then by assumption
$\CSequence'(\init)_b = \CSequenced'(\init)_b$.  Also, by the
induction hypothesis $\CSequence'(\init)_a = \CSequenced'(\init)_a$,
so by the form of $\Call$
\[
\CSequence(\init)_a =
\CSequence'.\Call(\init)_a =  \CSequence'(\init)_a \cup  \CSequence'(\init)_b = 
\CSequenced'(\init)_a \cup  \CSequenced'(\init)_b = \CSequenced'.\Call(\init)_a =
\CSequenced(\init)_a.
\]
\HB
\end{proof}

\begin{corollary} \label{cor:Ka} For all call types $\tau$, agents
  $a, b$ and call sequences $\CSequence$
\[
\mbox{$\CSequence \models^\calltype K_a F_a B$ iff $\CSequence \models^\calltype F_a B$.}
\]
\end{corollary}

\begin{proof}
By Proposition \ref{prop:a} and
the definition of truth of $K_a F_a B$ and $F_a B$.
\qed
\end{proof}


\section{Classification of the $\sim^\calltype_a$ Relations}
\label{sec:comparison}

We introduced in the previous section 18 equivalence relations
$\sim^\calltype_a$, each parame\-trised by an agent $a$. The uniform
presentation makes it possible to compare these relations by means of
a classification, which we now provide.

First, let us introduce some notation.  Given two call types $\tau_1$
and $\tau_2$ we abbreviate the statement
$\forall a \in \Agents, \sim^{\tau_1}_a \subset \sim^{\tau_2}_a$ to
$\tau_1 \subset \tau_2$ and similarly for $\tau_1 \sse \tau_2$ and
$\tau_1 = \tau_2$.  Such statements presuppose that we systematically
change the types of all calls in the considered call sequences.

The following theorem provides the announced classification.  It
clarifies in total 153 (= $\frac{18 \cdot 17}{2}$) relationships between the
equivalence relations.

\begin{theorem} \label{thm:classification} 
The $\sim^\calltype_a$
  equivalence relations form preorders presented in Figures
  \ref{fig:1} and \ref{fig:2}.  An arrow $\to$ from $\tau_1$ to
  $\tau_2$ stands here for $\tau_1 \subset \tau_2$,
  $(\myprivacy{1}, \dirtype, \otype)$ for the set of six call types
  with the privacy degree $\myprivacy{1}$ that are all equal, and
  $(\myprivacy{2}, \pushpull, \otype)$ for the set
  $\{(\myprivacy{2}, \pushpull, \alpha), (\myprivacy{2}, \pushpull,
  \beta)\}$.

\end{theorem}

\begin{figure}[htb]

\begin{center}
\begin{tikzpicture}[scale = 0.5]

\node (1) at (2,0) {$(\myprivacy{1}, \dirtype, \otype)$};
\node (2) at (2,2) {$(\myprivacy{2}, \pushpull, \otype)$};
\node (3) at (2,4) {$(\myprivacy{3}, \pushpull, \beta)$};
\node (4) at (2,6) {$(\myprivacy{3}, \pushpull, \alpha)$};

  \draw[->,thick] (1) -- (2);
  \draw[->,thick] (2) -- (3);
  \draw[->,thick] (3) -- (4);

\node (12) at (-5,2) {$(\myprivacy{2}, \push, \beta)$};
\node (13) at (-7,4) {$(\myprivacy{3}, \push, \beta)$};
\node (14) at (-3,4) {$(\myprivacy{2}, \push, \alpha)$};
\node (15) at (-5,6) {$(\myprivacy{3}, \push, \alpha)$};

  \draw[->,thick] (1) -- (12);
  \draw[->,thick] (12) -- (2);
  \draw[->,thick] (14) -- (2);
  \draw[->,thick] (12) -- (13);
  \draw[->,thick] (12) -- (14);
  \draw[->,thick] (13) -- (15);
  \draw[->,thick] (14) -- (15);

\node (22) at (9,2) {$(\myprivacy{2}, \pull, \beta)$};
\node (23) at (7,4) {$(\myprivacy{3}, \pull, \beta)$};
\node (24) at (11,4) {$(\myprivacy{2}, \pull, \alpha)$};
\node (25) at (9,6) {$(\myprivacy{3}, \pull, \alpha)$};

  \draw[->,thick] (1) -- (22);
  \draw[->,thick] (22) -- (2);
  \draw[->,thick] (24) -- (2);
  \draw[->,thick] (22) -- (23);
  \draw[->,thick] (22) -- (24);
  \draw[->,thick] (23) -- (25);
  \draw[->,thick] (24) -- (25);

\end{tikzpicture}
\caption{Classification of the $\sim^\calltype_a$ relations when $|\Agents| = 3$.\label{fig:1}}
\end{center}
\end{figure}

\begin{figure}[htb]

\begin{center}
\begin{tikzpicture}[scale = 0.5]

\node (1) at (2,0) {$(\myprivacy{1}, \dirtype, \otype)$};
\node (2) at (2,2) {$(\myprivacy{2}, \pushpull, \beta)$};
\node (3) at (0,4) {$(\myprivacy{3}, \pushpull, \beta)$};
\node (4) at (4,4) {$(\myprivacy{2}, \pushpull, \alpha)$};
\node (5) at (2,6) {$(\myprivacy{3}, \pushpull, \alpha)$};

  \draw[->,thick] (1) -- (2);
  \draw[->,thick] (2) -- (3);
  \draw[->,thick] (2) -- (4);
  \draw[->,thick] (3) -- (5);
  \draw[->,thick] (4) -- (5);

\node (12) at (-5,2) {$(\myprivacy{2}, \push, \beta)$};
\node (13) at (-7,4) {$(\myprivacy{3}, \push, \beta)$};
\node (14) at (-3,4) {$(\myprivacy{2}, \push, \alpha)$};
\node (15) at (-5,6) {$(\myprivacy{3}, \push, \alpha)$};

  \draw[->,thick] (1) -- (12);
  \draw[->,thick] (12) -- (13);
  \draw[->,thick] (12) -- (14);
  \draw[->,thick] (13) -- (15);
  \draw[->,thick] (14) -- (15);

\node (22) at (9,2) {$(\myprivacy{2}, \pull, \beta)$};
\node (23) at (7,4) {$(\myprivacy{3}, \pull, \beta)$};
\node (24) at (11,4) {$(\myprivacy{2}, \pull, \alpha)$};
\node (25) at (9,6) {$(\myprivacy{3}, \pull, \alpha)$};

  \draw[->,thick] (1) -- (22);
  \draw[->,thick] (22) -- (23);
  \draw[->,thick] (22) -- (24);
  \draw[->,thick] (23) -- (25);
  \draw[->,thick] (24) -- (25);

\end{tikzpicture}
\caption{Classification of the $\sim^\calltype_a$ relations
 when $|\Agents| > 3$.\label{fig:2}}
\end{center}
\end{figure}

The remainder of this section is devoted to the proof of Theorem
\ref{thm:classification}.  Below we say that the call types $\tau_1$
and $\tau_2$ are \bfe{incomparable} when neither $\tau_1 \sse \tau_2$
nor $\tau_2 \sse \tau_1$ holds.  The proofs concerning the
incomparability that are established below also hold for a stronger
definition, namely that $\tau_1$ and $\tau_2$ are incomparable when
for all agents $a$ neither $\sim^{\tau_1}_a \sse \sim^{\tau_2}_a$ nor
$\sim^{\tau_2}_a \sse \sim^{\tau_1}_a$ holds.  This way Figures
\ref{fig:1} and \ref{fig:2} can be alternatively interpreted as
preorders on the $\sim^{\tau}_a$ equivalence relations, for any agent
$a$, where an arrow $\to$ from $\tau_1$ to $\tau_2$ stands then for
$\sim^{\tau_1}_a \subset \sim^{\tau_2}_a$.


We first establish the claimed equalities between the call types.

\begin{lemma} \label{pro:1}
\mbox{} 
\begin{enumerate}[(i)]

\item 
Suppose that $\tau(\prtype) = \myprivacy{1}$. Then each
$\sim^{\tau}_a$ is the identity relation.

\item Suppose that $\tau_1(\prtype) = \tau_2(\prtype) = \myprivacy{1}$. 
Then $\tau_1 = \tau_2$. 

  \item If $|\Agents| = 3$ then 
$(\myprivacy{2}, \pushpull, \beta) = (\myprivacy{2}, \pushpull, \alpha)$.

\end{enumerate}
\end{lemma}

\begin{proof}
\mbox{}

\NI 
$(i)$ By Note \ref{not:1} it is sufficient to prove that
$\CSequence \approx^{\tau}_a \CSequenced$ implies
$\CSequence = \CSequenced$.  We proceed by induction on the sum $k$ of
the lengths $|\CSequence| + |\CSequenced|$ of both sequences. If
$k = 0$, then $\CSequence = \CSequenced = \epsilon$, so the claim
holds.  Suppose the claim holds for all pairs of sequences such that
the sum of their lengths is $< k$ and that $k > 0$,
$|\CSequence| + |\CSequenced| = k$ and
$\CSequence \approx^{\tau}_a \CSequenced$.

  By definition $\approx^{\tau}_a$ is the smallest relation
  satisfying the Base and Step conditions of Definition \ref{def:aux}.
  So, since $\tau(\prtype) = \myprivacy{1}$, by the Step condition
  $\CSequence$ is of the form $\CSequence'.\Call$ and $\CSequenced$ is
  of the form $\CSequenced'.\Call$, where
  $\CSequence' \approx^{\tau}_a \CSequenced'$.  By the induction
  hypothesis $\CSequence' = \CSequenced'$, so $\CSequence = \CSequenced$.  

\smallskip

\noindent
$(ii)$ 
By $(i)$.
\smallskip

\NI 
$(iii)$ Suppose $\Agents = \{a, b, c\}$.  Take
$\tau \in \{(\myprivacy{2}, \pushpull, \beta),
(\myprivacy{2},\pushpull, \alpha)\}$.  Then by Definition
\ref{def:aux} $\CSequence \sim_a^{\tau} \CSequenced$ iff $\CSequence$
and $\CSequenced$ differ only in some of the calls $a$ is not involved
in.  Because there are exactly 3 agents, each such call must be
$b \pushpull c$ or $c \pushpull b$ and both have the same effect
independently of the type of observance.
\HB
\end{proof}

Next we establish the claimed strict inclusions.  Below the
unspecified parameters are implicitly universally qualified.  For
example,
$(\myprivacy{1}, \dirtype, \otype) \subset (\myprivacy{2}, \dirtype,
\otype)$ is an abbreviation for the statement
\[
  \forall a \in \Agents \ \forall \dirtype \in \dirset \ \forall \otype \in \oset \sim^{(\myprivacy{1}, \dirtype,
    \otype)}_a \ \subset \ \sim^{(\myprivacy{2}, \dirtype, \otype)}_a.
\]

\begin{lemma} \label{pro:comparison} 
\mbox{} 
  \begin{enumerate}[(i)]

  \item 
    $(\myprivacy{1}, \dirtype, \otype) \subset (\myprivacy{2},
    \dirtype, \otype)$.

  \item $(\myprivacy{2}, \dirtype, \otype) \subset (\myprivacy{3}, \dirtype, \otype)$. 

  \item If $|\Agents| > 3$ or $\dirtype \neq \pushpull$ then
$(\myprivacy{2}, \dirtype, \beta) \subset (\myprivacy{2}, \dirtype, \alpha)$.

\item 
If $|\Agents| = 3$, $\dirtype \neq \pushpull$ and  $\otype_1, \otype_2 \in \oset$,
then $(\myprivacy{2}, \dirtype, \otype_1) \subset (\myprivacy{2}, \pushpull, \otype_2)$.

  \item $(\myprivacy{3}, \dirtype, \beta) \subset (\myprivacy{3}, \dirtype, \alpha)$. 

  \end{enumerate}
\end{lemma}
\begin{proof}
First we establish the $\sse$ inclusions.
\smallskip

\NI
$(i)$ and $(ii)$ 
These are direct consequences of Definition \ref{def:aux}.
\smallskip

\NI $(iii)$ We prove that
$(\myprivacy{2}, \dirtype, \beta) \sse (\myprivacy{2}, \dirtype,
\alpha)$ always holds.  Let
$\tau_1 = (\myprivacy{2}, \dirtype, \beta)$ and
$\tau_2 = (\myprivacy{2}, \dirtype, \alpha)$.  Fix an agent $a$.

By Note \ref{not:1} it is sufficient to prove that
$\CSequence \approx^{\tau_1}_a \CSequenced$ implies
$\CSequence \approx^{\tau_2}_a \CSequenced$. We proceed by induction
on the sum $k$ of the lengths $|\CSequence| + |\CSequenced|$ of both
sequences. If $k = 0$, then $\CSequence = \CSequenced = \epsilon$, so
the claim holds.  Suppose the claim holds for all pairs of sequences
such that the sum of their lengths is $< k$ and that $k > 0$,
$|\CSequence| + |\CSequenced| = k$ and
$\CSequence \approx^{\tau_1}_a \CSequenced$.  By definition
$\approx^{\tau_1}_a$ is the smallest relation satisfying the Base and
Step conditions of Definition \ref{def:aux}.  So, since
$\tau_1(\prtype) = \myprivacy{2}$, by the Step condition $\CSequence$
is of the form $\CSequence'.\Call$ and $\CSequenced$ is of the form
$\CSequenced'.\Calld$, where
$\CSequence' \approx^{\tau_1}_a \CSequenced'$.  By the induction
hypothesis $\CSequence' \approx^{\tau_2}_a \CSequenced'$.

Three cases arise that reflect the case analysis in Definition \ref{def:aux},
where $b$ is the partner of $a$ in the call $\Call$:

\begin{enumerate}[(a)]

\item $a \not\in \Call, a \not\in \Calld$.

Then $\CSequence' \approx^{\tau_2}_a \CSequenced'$ implies
$\CSequence'.\Call \approx^{\tau_2}_a \CSequenced'.\Calld$.

\item  $\Call \in \{a \push b, b\, \pull \, a\}$.

Then $\CSequence'.\Call \approx^{\tau_1}_a \CSequenced'.\Calld$ implies
$\Call = \Calld$ and consequently $\CSequence' \approx^{\tau_2}_a \CSequenced'$
implies $\CSequence'.\Call \approx^{\tau_2}_a \CSequenced'.\Calld$.

\item $\Call \in \{a \pushpull b, b \pushpull a, b \push a, a\, \pull \, b\}$.

  Then $\CSequence'.\Call \approx^{\tau_1}_a \CSequenced'.\Calld$
  implies $\Call = \Calld$ and
  $\CSequence'(\init)_b = \CSequenced'(\init)_b$, because
  $\tau_1(\otype) = \beta$.  Also by Proposition \ref{prop:a}
  $\CSequence'(\init)_a = \CSequenced'(\init)_a$, so
\[
\CSequence'.\Call(\init)_a =  \CSequence'(\init)_a \cup  \CSequence'(\init)_b = 
\CSequenced'(\init)_a \cup  \CSequenced'(\init)_b = \CSequenced'.\Call(\init)_a.
\]

Hence $\CSequence' \approx^{\tau_2}_a \CSequenced'$ implies 
$\CSequence'.\Call \approx^{\tau_2}_a \CSequenced'.\Calld$, because $\tau_2(\otype) = \alpha$. 
\end{enumerate}

\NI
$(iv)$
Let $\tau_1 = (\myprivacy{2}, \dirtype, \otype_1)$ and
$\tau_2 = (\myprivacy{2}, \pushpull, \otype_2)$.  Fix an agent $a$.

By Note \ref{not:1} it suffices to prove that
$\CSequence \approx^{\tau_1}_a \CSequenced$ implies
$\CSequence \approx^{\tau_2}_a \CSequenced$.  Two cases arise.
\begin{enumerate}[(a)]
\item  $\dirtype = \push$.

  Because there are only 3 agents, by Definition \ref{def:aux} if
  $\CSequence \approx^{\tau_1}_a \CSequenced$ then $\CSequence$ and
  $\CSequenced$ differ only in some of the calls $a$ is not involved
  in.  So then $\CSequence$ and $\CSequenced$, when interpreted under
  $\tau_2$, differ only in some of the calls between $b$ and $c$, which
  are $b \pushpull c$ or $c \pushpull b$, and both have the same effect
  independently of the type of observance.

So $\CSequence \sim_a^{\tau_1} \CSequenced$ implies
$\CSequence \sim_a^{\tau_2} \CSequenced$.

\item  $\dirtype = \pull$.

The argument is the same as in (a). 
\end{enumerate}

\NI
$(v)$ The proof is analogous to the one given in $(iii)$ and is omitted.
\smallskip

We prove now that the inclusions are strict.
Suppose that $a, b, c \in \Agents$ are different agents.
\smallskip

\NI
$(i)$
Note that for
$\tau_1 = (\myprivacy{1}, \dirtype, \otype)$ and $\tau_2 = (\myprivacy{2}, \dirtype, \otype)$
we have $bc \sim_a^{\tau_2} cb$, while $bc \not\sim_a^{\tau_1} cb$.



\smallskip

\NI
$(ii)$ Note that for $\tau_1 = (\myprivacy{2}, \dirtype, \otype)$ and
$\tau_2 = (\myprivacy{3}, \dirtype, \otype)$ we have
$bc \sim_a^{\tau_2} \epsilon$, while
$bc \not \sim_a^{\tau_1} \epsilon$.
\smallskip

\NI
$(iii)$ 
Let $\tau_1 = (\myprivacy{2}, \dirtype, \beta)$ and
$\tau_2 = (\myprivacy{2}, \dirtype, \alpha)$. 
Assume first that
$|\Agents| > 3$. Suppose $a, b, c, d \in \Agents$  are different agents.
Three cases arise.

\begin{enumerate}[(a)]
\item  $\dirtype = \pushpull$.

  Then
  $a \pushpull b, a \pushpull c, b \pushpull c, a \pushpull b
  \sim_a^{\tau_2} a \pushpull b, a \pushpull c, c \pushpull d, a
  \pushpull b$, while
  $a \pushpull b, a \pushpull c, b \pushpull c, a \pushpull b \not
  \sim_a^{\tau_1} a \pushpull b, a \pushpull c, c \pushpull d, a
  \pushpull b$.

\item  $\dirtype = \push$.

Then
$c \push a, b \push c, b \push a \sim_a^{\tau_2} c \push a, c \push b,
b \push a$, while
$c \push a, b \push c, b \push a \not \sim_a^{\tau_1} c \push a, c
\push b, b \push a$.

\item  $\dirtype = \pull$.

  $a\, \pull\, c, c\, \pull\, b, a\, \pull\, b \sim_a^{\tau_2} a\,
  \pull\, c, b\, \pull\, c, a\, \pull\, b$, while
  $a\, \pull\, c, c\, \pull\, b, a\, \pull\, b \not \sim_a^{\tau_1}
  a\, \pull\, c, b\, \pull\, c, a\, \pull\, b$.

\end{enumerate}

Assume now that $\dirtype \neq \pushpull$.  Then the desired
conclusion is established in (b) and (c), as the examples used there
involve only three agents.

\smallskip

\NI
$(iv)$ Let $\tau_1 = (\myprivacy{2}, \dirtype, \otype_1)$ and
$\tau_2 = (\myprivacy{2}, \pushpull, \otype_2)$. Assume
$\Agents = \{a, b, c\}$.
Two cases arise.
\begin{enumerate}[(a)]
\item  $\dirtype = \push$.


Note that
$c \pushpull b, c \pushpull a \sim_a^{\tau_2} b \pushpull c, c
\pushpull a$, while
$c \push b, c \push a \not \sim_a^{\tau_1} b \push c, c \push a$.

\item  $\dirtype = \pull$.


Note that
$b \pushpull c, a \pushpull c \sim_a^{\tau_2} c \pushpull b, a \pushpull c$, while
$b\, \pull\, c, a\, \pull\, c \not \sim_a^{\tau_1} c\, \pull\, b, a\, \pull\, c$.

\end{enumerate}

\NI
$(v)$ 
Let $\tau_1 = (\myprivacy{3}, \dirtype, \alpha)$
and $\tau_2 = (\myprivacy{3}, \dirtype, \beta)$.
Three cases arise.

\begin{enumerate}[(a)]
\item  $\dirtype = \pushpull$.

  Note that
  $a \pushpull c, a \pushpull b \sim_a^{\tau_1} a \pushpull c, b
  \pushpull c, a \pushpull b$, while
  $a \pushpull c, a \pushpull b \not\sim_a^{\tau_2} a \pushpull c, b
  \pushpull c, a \pushpull b$.

\item  $\dirtype = \push$.

Note that
  $c \push a, b \push a \sim_a^{\tau_1} c \push a, c \push b, b \push
  a $, while
  $c \push a, b \push a \not \sim_a^{\tau_2} c \push a, c \push b, b
  \push a$.

\item  $\dirtype = \pull$.

Note that
  $a\, \pull\, c, a\, \pull\, b \sim_a^{\tau_1} a\, \pull\, c, b \,
  \pull\, c, a\, \pull\, b$, while
  $a\, \pull\, c, a\, \pull\, b \not\sim_a^{\tau_2} a\, \pull\, c, b
  \, \pull\, c, a\, \pull\, b$.
\end{enumerate}
\HB
\end{proof}

As a side remark note that Lemmas \ref{pro:1}$(ii), (iii)$ and
\ref{pro:comparison}$(iii), (v)$ imply that
\[
(\prtype, \dirtype, \beta) \sse (\prtype, \dirtype, \alpha).
\]

Finally, we establish the claims concerning incomparability of the types.

\begin{lemma} \label{pro:2} 
Let $\dirtype, \dirtype_1, \dirtype_2 \in \dirset$ and $\otype_1, \otype_2 \in \oset$.
  \begin{enumerate}[(i)]

  \item Suppose that $|\Agents| > 3$ or
$\pushpull \not\in \{\dirtype_1, \dirtype_2\}$, and
$\dirtype_1 \neq \dirtype_2$.  Then
$(\myprivacy{2}, \dirtype_1, \otype_1)$ and
$(\myprivacy{2}, \dirtype_2, \otype_2)$ are incomparable.

\item 
Suppose that $\dirtype_1 \neq \dirtype_2$. Then
$(\myprivacy{3}, \dirtype_1, \otype_1)$ and
$(\myprivacy{3}, \dirtype_2, \otype_2)$ are incomparable.

\item 
Suppose that $|\Agents| = 3$ and $\dirtype \neq \pushpull$. 
Then $(\myprivacy{2}, \pushpull, \alpha)$ and
$(\myprivacy{3}, \dirtype, \alpha)$ are incomparable.

\item 
Suppose that $|\Agents| > 3$ or $\pushpull \not\in \{\dirtype_1, \dirtype_2\}$,
and $\dirtype_1 \neq \dirtype_2$. 
Then $(\myprivacy{2}, \dirtype_1, \beta)$ and
$(\myprivacy{3}, \dirtype_2, \alpha)$ are incomparable.

\item 
Suppose that $|\Agents| > 3$ or $\dirtype \neq \pushpull$.  Then
  $(\myprivacy{2}, \dirtype, \alpha)$ and
  $(\myprivacy{3}, \dirtype, \beta)$ are incomparable.

  \end{enumerate}
\end{lemma}













\begin{proof}
Suppose that $a, b, c \in \Agents$ are different agents.

\NI
$(i)$
Let $\tau_1 = (\myprivacy{2}, \dirtype_1, \otype_1)$ and
$\tau_2 = (\myprivacy{2}, \dirtype_2, \otype_2)$.

Assume first that $|\Agents| > 3$. Suppose $a, b, c, d \in \Agents$ are different agents.
For each pair of distinct direction types we exhibit appropriate call sequences. In each case
the conclusions do not depend on the observance level.
\begin{enumerate}[(a)]

\item \label{1}
$\dirtype_1 = \pull$ and 
$\dirtype_2 = \pushpull$.

Then $b\, \pull \, c, c\, \pull \, a \sim_a^{\tau_1} b\, \pull \, d, c\, \pull \, a$, while
$b \pushpull c, c \pushpull a \not\sim_a^{\tau_2} b \pushpull d, c \pushpull a$.



Further, $b \pushpull c, a \pushpull c \sim_a^{\tau_2} c \pushpull b, a \pushpull c$, while
$b \, \pull \, c, a \, \pull \, c \not\sim_a^{\tau_1} c \, \pull \, b, a \, \pull \, c$.

\item \label{2}
 $\dirtype_1 = \push$ and 
$\dirtype_2 = \pushpull$.

Then  $c \push b, a \push c \sim_a^{\tau_1} d \push b, a \push c$, while
$c \pushpull b, a \pushpull c \not\sim_a^{\tau_2} d \pushpull b, a \pushpull c$.



Further, $c \pushpull b, c \pushpull a \sim_a^{\tau_2} b \pushpull c, c \pushpull a$, while
$c \push b, c \push a \not \sim_a^{\tau_1} b \push c, c \push a$.

\item \label{3}
 $\dirtype_1 = \push$ and 
$\dirtype_2 = \pull$.


Then $b \push c, a \push c \sim_a^{\tau_1} c \push b, a \push c$, while
$b\, \pull \, c, a \, \pull \, c \not \sim_a^{\tau_2} c \, \pull \, b, a \, \pull \, c$.

Further,
$c\, \pull \, b, c\, \pull \, a \sim_a^{\tau_2} b\, \pull \, c, c\, \pull \, a$,
while
$c \push b, c \push a \not \sim_a^{\tau_1} b \push c, c \push a$.  
\end{enumerate}

Assume now that
$\pushpull \not\in \{\dirtype_1, \dirtype_2\}$.  Then the
desired conclusion is established in (\ref{3}), as both examples used
there involve only three agents.

\smallskip

\NI
$(ii)$  
Let $\tau_1 = (\myprivacy{3}, \dirtype_1, \otype_1)$ and
$\tau_2 = (\myprivacy{3}, \dirtype_2, \otype_2)$.
We proceed by the same case analysis as in the proof of $(i)$.

\begin{enumerate}[(a)]
\item 
$\dirtype_1 = \pull$ and $\dirtype_2 = \pushpull$.

Then $b\, \pull \, c, c\, \pull \, a \sim_a^{\tau_1} c\, \pull \, a$, while
$b \pushpull c, c \pushpull a \not\sim_a^{\tau_2} c \pushpull a$.

\item 
$\dirtype_1 = \push$ and $\dirtype_2 = \pushpull$.

Then  $c \push b, a \push c \sim_a^{\tau_1} a \push c$, while
$c \pushpull b, a \pushpull c \not\sim_a^{\tau_2} a \pushpull c$.

\item 
 $\dirtype_1 = \push$ and $\dirtype_2 = \pull$.

Then both examples used in the proof of item (\ref{3}) in $(i)$ apply here, as well.
\end{enumerate}

To prove that  $\CSequence \sim_a^{\tau_2} \CSequenced$ does not imply
$\CSequence \sim_a^{\tau_1} \CSequenced$ we can use the same examples as in the proof of $(i)$.
\smallskip

\NI
$(iii)$
Let 
$\tau_1 = (\myprivacy{2}, \pushpull, \alpha)$ and
$\tau_2 = (\myprivacy{3}, \dirtype, \alpha)$.
We distinguish two cases.
\begin{enumerate}[(a)]

\item $\dirtype = \push$.

Then $c \pushpull b, c \pushpull a \sim_a^{\tau_1} b \pushpull c, c \pushpull a$, while
$c \push b, c \push a \not \sim_a^{\tau_2} b \push c, c \push a$.

\item $\dirtype = \pull$.

Then $b\pushpull c, a \pushpull c \sim_a^{\tau_1} c \pushpull b, a \pushpull c$, while
$b \, \pull\, c, a \, \pull\, c \not \sim_a^{\tau_2} c \, \pull\, b, a \, \pull\, c$.
\end{enumerate}
Next, note that for all $\dirtype \neq \pushpull$ we have
$bc \sim_a^{\tau_2} \epsilon$, while
$bc \not \sim_a^{\tau_1} \epsilon$.
\smallskip

\NI
$(iv)$
Let $\tau_1 = (\myprivacy{2}, \dirtype_1, \beta)$ and
$\tau_2 = (\myprivacy{3}, \dirtype_2, \alpha)$.

Assume first that $|\Agents| > 3$. Suppose $a, b, c, d \in \Agents$ are different agents.
For each pair of distinct direction types we exhibit appropriate call sequences. 
\begin{enumerate} [(a)]

\item $\dirtype_1 = \push$, $\dirtype_2 = \pull$.

Then $c \push b, a \push b \sim_a^{\tau_1} b \push c, a \push b$, while
$c\, \pull \, b, a\, \pull \, b \not \sim_a^{\tau_2} b\, \pull \, c, a\, \pull \, b$.

\item $\dirtype_1 = \push$, $\dirtype_2 = \pushpull$.


Then $c \push b, a \push c \sim_a^{\tau_1} d \push b, a \push c$, while
$c \pushpull b, a \pushpull c \not\sim_a^{\tau_2} d \pushpull b, a \pushpull c$.

\item $\dirtype_1 = \pull$, $\dirtype_2 = \push$.

Then $b\, \pull \, c, b\, \pull \, a \sim_a^{\tau_1} c\, \pull \, b, b\, \pull \, a$, while
$b \push c, b \push a \not \sim_a^{\tau_2} c \push b, b \push a$.

\item $\dirtype_1 = \pull$, $\dirtype_2 = \pushpull$.

Then $b\, \pull \, c, c\, \pull \, a \sim_a^{\tau_1} b\, \pull \, d, c\, \pull \, a$, while
$b \pushpull c, c \pushpull a \not\sim_a^{\tau_2} b \pushpull d, c \pushpull a$.

\item $\dirtype_1 = \pushpull$, $\dirtype_2 = \push$.

Then $c \pushpull b, c \pushpull a \sim_a^{\tau_1} b \pushpull c, c \pushpull a$, while
$c \push b, c \push a \not \sim_a^{\tau_2} b \push c, c \push a$.

\item $\dirtype_1 = \pushpull$, $\dirtype_2 = \pull$.

Then $b \pushpull c, a \pushpull c \sim_a^{\tau_1} c \pushpull b, a \pushpull c$, while
$b \, \pull \, c, a \, \pull \, c \not\sim_a^{\tau_2} c \, \pull \, b, a \, \pull \, c$.

\end{enumerate}

Assume now that
$\pushpull \not\in \{\dirtype_1, \dirtype_2\}$.  Then the
desired conclusion is established in (a) and (c), as both examples used
there involve only three agents.

Next, note that for all direction types
$bc \sim_a^{\tau_2} \epsilon$, while
$bc \not \sim_a^{\tau_1} \epsilon$.
\smallskip

\NI
$(v)$ Let $\tau_1 = (\myprivacy{2}, \dirtype, \alpha)$ and
$\tau_2 = (\myprivacy{3}, \dirtype, \beta)$.

Assume first that $|\Agents| > 3$. Suppose $a, b, c, d \in \Agents$
are different agents.  

\begin{enumerate}[(a)]

\item $\dirtype = \pushpull$.

Then 
$a \pushpull b, a \pushpull c, b \pushpull c, a \pushpull b
\sim_a^{\tau_1} a \pushpull b, a \pushpull c, c \pushpull d, a
\pushpull b$, while
$a \pushpull b, a \pushpull c, b \pushpull c, a \pushpull b \not
\sim_a^{\tau_2} a \pushpull b, a \pushpull c, c \pushpull d, a
\pushpull b$.

\item $\dirtype = \push$.

  Then
  $c \push a, b \push c, b \push a \sim_a^{\tau_1} c \push a, c \push
  b, b \push a$, while
  $c \push a, b \push c, b \push a \not \sim_a^{\tau_2} c \push a, c
  \push b, b \push a$.

\item $\dirtype = \pull$.

  Then
  $a\, \pull\, c, c\, \pull\, b, a\, \pull\, b \sim_a^{\tau_1} a\, \pull\, c, b\, \pull\,
  c, a\, \pull\, b$, while
  $a\, \pull\, c, c\, \pull\, b, a\, \pull\, b \not \sim_a^{\tau_2} a\, \pull\, c, b\, \pull\,
  c, a\, \pull\, b$.
\end{enumerate}

Assume now that $\dirtype \neq \pushpull$.  Then the desired
conclusion is established in (b) and (c) as both examples used there
involve only three agents.

Finally, note that for all direction types
$bc \sim_a^{\tau_2} \epsilon$, while
$bc \not \sim_a^{\tau_1} \epsilon$.
\HB
\end{proof}



The above Lemmas imply the classification of the
$\sim^\calltype_a$ relations given in Theorem \ref{thm:classification}
and visualized in Figures \ref{fig:1} and \ref{fig:2}.  Indeed, the
equalities (represented as sets) are established in Lemma
\ref{pro:1}, the strict inclusions (that correspond to the arrows) are
established in Lemma \ref{pro:comparison}, and Lemma
\ref{pro:2} implies that no further strict inclusions (i.e., arrows)
are present. For example, there is no arrow in Figure \ref{fig:2}
between two different diamond shaped subgraphs that correspond to the
direction types $\push, \pushpull$, and $\pull$ because by Lemma
\ref{pro:2}$(iv)$ for $\dirtype_1 \neq \dirtype_2$ the call types
$(\myprivacy{2}, \dirtype_1, \beta)$ and $(\myprivacy{3}, \dirtype_2, \alpha)$ are
incomparable.

\section{Applications of the Classification} \label{sec:applications}

The section shows how the above systematisation of
$\sim^{\calltype}_a$ relations, through the standard epistemic logic
semantics of Definition \ref{def:truth}, enables general insights into
the epistemic effects of call sequences and offers a natural handle on
how to model assumptions to the effect that agents have common
knowledge of the protocol in use.


\subsection{Epistemic effects of communication types}
\label{sec:epistemic}

The above classification is useful in order to draw general epistemic
consequences in presence of different communication types.  
Below we will be using two fragments of $\lang$: 
\begin{itemize}
\item $\lang_1^+$, consisting of the {\em literals} $F_a S$ and
 $\neg F_a S$, $\wedge, \lor$ and $K_a$,

\item $\lang_2^+$, consisting of the atomic formulas $F_a S$, $\wedge, \lor$ and $K_a$.

\end{itemize}

\begin{proposition} \label{prop:1}
Consider two call types $\tau_1$ and $\tau_2$ such
  that $\tau_1(\dirtype) = \tau_2(\dirtype)$.

\begin{enumerate}[(i)]
\item 
For all literals $\psi$ and all  $\CSequence$, \ $\CSequence \models^{\tau_2} \psi \implies \CSequence \models^{\tau_1} \psi$.

\item If $\tau_1 \sse \tau_2$ then
\[
\mbox{for all formulas $\phi \in \lang^+_1$ and all  $\CSequence$,}\  \CSequence \models^{\tau_2} \phi \implies \CSequence \models^{\tau_1} \phi.
\]
\end{enumerate}
\end{proposition}

\begin{proof}
\mbox{}

\NI
$(i)$
By assumption $\tau_1(\dirtype) = \tau_2(\dirtype)$, so both occurrences of
$\CSequence$ refer to identical call sequences. Hence for all atomic formulas $F_a S$
and all  $\CSequence$, \ $\CSequence \models^{\tau_2} F_a S$ iff $\CSequence \models^{\tau_1} F_a S$.
\smallskip

\NI
$(ii)$
We proceed by induction on the structure of $\phi$. 
The only case that  requires explanation is when $\phi$ is
of the form $K_a \psi$. Suppose that
$\CSequence \models^{\tau_2} K_a \psi$.  To prove
$\CSequence \models^{\tau_1} K_a \psi$ take a call sequence
$\CSequenced$ such that $\CSequence \sim^{\tau_1}_a \CSequenced$.  By
assumption $\tau_1 \subseteq \tau_2$, hence
$\CSequence \sim^{\tau_2}_a \CSequenced$ and so
$\CSequenced \models^{\tau_2} \psi$.  By the induction hypothesis
$\CSequenced \models^{\tau_1} \psi$, so by definition
$\CSequence \models^{\tau_1} K_a \psi$.
\HB
\end{proof}

It is easy to construct examples showing that the implication in
$(ii)$ does not hold for all formulas.  For instance, for
$\tau_1 = (\myprivacy{1}, \pushpull, \alpha)$ and
$\tau_2 = (\myprivacy{3}, \pushpull, \alpha)$ we have
$\tau_1 \sse \tau_2$ by Theorem \ref{thm:classification} and
$bc \models^{\tau_2} \neg K_a F_b C$ but not $bc \models^{\tau_1} \neg K_a F_b C$.

We finally compare knowledge for call types with different
direction types.  Then claim $(i)$ in the above Proposition does not
hold anymore. Indeed, for $\tau_1$ and $\tau_2$ such that
$\tau_1(\dirtype) = \pushpull$ and $\tau_2(\dirtype) = \push$ we
have $ab \models^{\tau_2} \neg F_a B$ but not
$ab \models^{\tau_1} \neg F_a B$.
However, the following weaker claim does hold.

\begin{proposition}\label{prop:newnew}
Consider two call types $\tau_1$ and $\tau_2$ such
that $\tau_1(\dirtype) = \pushpull$. 

\begin{enumerate}[(i)]
\item 
For all atomic formulas $\psi$ and all  $\CSequence$, \ $\CSequence \models^{\tau_2} \psi \implies \CSequence \models^{\tau_1} \psi$.

\item If $\tau_1 \sse \tau_2$ then
\[
\mbox{for all formulas $\phi \in \lang^+_2$ and all  $\CSequence$,}\  \CSequence \models^{\tau_2} \phi \implies \CSequence \models^{\tau_1} \phi.
\]
\end{enumerate}
\end{proposition}

\begin{proof} 
By Proposition \ref{prop:1} we can assume that  $\tau_2(\dirtype) \neq \pushpull$.

\NI
$(i)$
We use induction on the
length $|\cc|$ of $\cc$.  Assume that
$\tau_2(\dirtype) = \push$.  
If $|\cc| = 0$ then $\cc = \epsilon$ and
$ \epsilon \models^{\tau_2} F_c D$ iff $D = C$
iff $ \epsilon \models^{\tau_1} F_c D$.
Now suppose the claim is proven for $\cc$ and consider $\cc. ab$. 

For any agent $c \neq b$, we have by Definition~\ref{def:effects}
$\cc. a\push b \models^{\tau_2} F_c D$ iff
$\cc \models^{\tau_2} F_c D$, which implies by the induction
hypothesis $\cc \models^{\tau_1} F_c D$, and hence
$\cc. a\pushpull b \models^{\tau_1} F_c D$.  For agent $b$, we have
$\cc. a \push b \models^{\tau_2} F_bD$ iff
($\cc \models^{\tau_2} F_aD$ or $\cc \models^{\tau_2} F_bD$) and
$\cc. a \pushpull b \models^{\tau_1} F_bD$ iff
($\cc\models^{\tau_1} F_aD$ or $\cc \models^{\tau_1} F_bD$), so the
claim for $b$ holds by the induction hypothesis, as well.

The proof for $\tau_2(\dirtype) = \pull$ is analogous and omitted.
\smallskip

\NI
$(ii)$
The claim follows by $(i)$ and the argument used in the proof of Proposition
\ref{prop:1}.
\HB
\end{proof}

Proposition \ref{prop:1} holds for example for
$\tau_1 = (\myprivacy{1}, \pushpull, \beta)$ and
$\tau_2 = (\myprivacy{3}, \pushpull, \alpha)$, since by Theorem \ref{thm:classification}
\[
(\myprivacy{1}, \pushpull, \beta) \subset
(\myprivacy{2}, \pushpull, \beta) \subset
(\myprivacy{3}, \pushpull, \beta) \subset
(\myprivacy{3}, \pushpull, \alpha).
\]

In turn, Proposition \ref{prop:newnew} holds for example for
$\tau_1 = (\myprivacy{1}, \pushpull, \beta)$ and
$\tau_2 = (\myprivacy{3}, \push, \alpha)$, since by Theorem \ref{thm:classification}
\[
(\myprivacy{1}, \pushpull, \beta) = 
(\myprivacy{1}, \push, \beta) \subset
(\myprivacy{2}, \push, \beta) \subset
(\myprivacy{3}, \push, \beta) \subset
(\myprivacy{3}, \push, \alpha).
\]

In particular, for both pairs of $\tau_1$ and $\tau_2$ for all call
sequences $\CSequence$, $\CSequence \models^{\tau_2} K_a F_b C$
implies $\CSequence \models^{\tau_1} K_a F_b C$.
Informally, under $\tau_1$ the agents are then more informed about the
knowledge of other agents than under $\tau_2$.

Further, note that by Theorem \ref{thm:classification} if
$\tau_1(\dirtype) = \pushpull \neq \tau_2(\dirtype)$ then
$\tau_1 \sse \tau_2$ iff $\tau_1(\prtype) = \myprivacy{1}$, so under
the assumption $\tau_1(\dirtype) = \pushpull \neq \tau_2(\dirtype)$
the second claim of Proposition \ref{prop:newnew} can be rewritten as
\[
\mbox{If $\tau_1(\prtype) = \myprivacy{1}$ then for all formulas $\phi \in \lang^+_2$ and all  $\CSequence$,}\  \CSequence \models^{\tau_2} \phi \implies \CSequence \models^{\tau_1} \phi.
\]

This implication (under the assumption
$\tau_1(\dirtype) = \pushpull \neq \tau_2(\dirtype)$) does not hold
for the other two privacy types because of the following instructive
counterexample.

\begin{example}
Assume $\Agents = \{a,b,c\}$.
Suppose $\tau_1 = (\prtype, \pushpull, \alpha)$,
$\tau_2 = (\prtype, \push, \alpha)$, where $\prtype \neq \myprivacy{1}$,
and $\CSequence = ac, cb, ba$.
We claim that then $\CSequence \models^{\tau_2} K_a K_c F_bC$ but not
$\CSequence \models^{\tau_1} K_a K_c F_bC$.
\smallskip

\NI
$(i)$ $\prtype = \myprivacy{2}$.

Then $\CSequence \models^{\tau_2} K_a K_c F_bC$. The reason is that
the only call sequence $\sim^{\tau_2}_a$ equivalent to $\CSequence$ is
$\CSequence$ itself.  Indeed, if
$\CSequence \sim^{\tau_2}_a \CSequenced$, then $\CSequenced$ has to be
of the form $ac, \Call, ba$, where $a \not\in \Call$, and by the first
entry in Table \ref{table:een}, bottom, also
$\CSequence(\init)_a = (ac, \Call, ba)(\init)_a$ has to hold. But
$\CSequence(\init)_a = \{A,B,C\}$, which implies that $\Call = cb$.

Thus $\CSequence \models^{\tau_2} K_a K_c F_bC$ iff
$\CSequence \models^{\tau_2} K_c F_bC$ and the latter is easy to
check.  However, $\CSequence \not\models^{\tau_1} K_a K_c F_bC$ since
$\CSequence \sim^{\tau_1}_a ac, de, ba$ and
$ac, de, ba \not \models^{\tau_1} K_c F_bC$.  \smallskip

\NI
$(ii)$ $\prtype = \myprivacy{3}$.

The reasoning is now a bit more involved.  To show that
$\CSequence \models^{\tau_2} K_a K_c F_bC$ take a call sequence
$\CSequenced$ such that $\CSequence \sim^{\tau_2}_a \CSequenced$.
Then $\CSequenced$ is of the form
$\CSequenced_1, ac, \CSequenced_2, ba, \CSequenced_3$, where agent $a$
is not involved in any call from
$\CSequenced_1, \CSequenced_2, \CSequenced_3$. Moreover
$\CSequence(\init)_a = \CSequenced(\init)_a$ holds, as well. But, as
already noted in $(i)$, $\CSequence(\init)_a = \{A,B,C\}$, which
implies that one of the calls in $\CSequenced_1$ or $\CSequenced_2$ is
$cb$.

Now, for any call sequence $\CSequenced$ in which the call $cb$
appears we have $\CSequenced \models^{\tau_2} K_c F_bC$. Indeed, if
$\CSequenced \sim^{\tau_2}_c \CSequenced'$ then the call $cb$ appears
in $\CSequenced'$, as well, and hence
$\CSequenced' \models^{\tau_2} F_bC$.
However, $\CSequence \not\models^{\tau_1} K_a K_c F_bC$ since, as in
$(i)$, $\CSequence \sim^{\tau_1}_a ac, de, ba$ and
$ac, de, ba \not \models^{\tau_1} K_c F_bC$.

Analogous examples can be constructed for the pull calls.  
\HB
\end{example}

This example shows that for the $\myprivacy{2}$ and $\myprivacy{3}$
privacy degrees the push calls may convey more knowledge than the
push-pull calls, even though the former ones result in less
informative communication.  The same is the case for the pull calls.




\subsection{Common knowledge of protocols} \label{sec:further}

When reasoning about specific protocols it is necessary to limit the
set of considered call sequences to those that are `legal' for it.
When the agents form a graph given in advance one can simply limit the
set of considered call sequences by allowing only syntactically legal
calls.  This affects the definition of semantics and can be of
importance when reasoning about the correctness of specific protocols.

For example, in \cite{apt15epistemic} a specific protocol for a
directed ring is proved correct (Protocol R2 on page 61, for 3 or 4
agents) by allowing for each agent $a$ only the calls between her and
her successor $a \oplus 1$, and using the fact that the formula
$K_a F_{a \oplus 1} A \ominus 1 \to F_a A \ominus 1$ is then true.
Here $A \ominus 1$ is the secret of the predecessor of agent $a$, so
this formula states that if agent $a$ knows that her successor is
familiar with the secret $A \ominus 1$ of her predecessor then agent
$a$ is familiar with the secret $A \ominus 1$.

A more challenging task is to incorporate into the framework an
assumption that the agents have common knowledge of the underlying
protocol.\footnote{This issue was identified as an open problem for
  epistemic gossip in \cite{apt15epistemic}. The same issue manifests
  itself in other knowledge-based asynchronous protocols, such as the
  one investigated recently in \cite{knight17reasoning}.}

\begin{example}
  Consider Protocol \ref{protocol.mineen} (Hear my Secret) from
  Section \ref{sec:preliminaries} with the direction type $\pushpull$.
  Recall that in this protocol an agent $a$ can call agent $b$ if
  $\neg K_a F_b a$ is true after the current call sequence.  So each
  pair of agents can communicate at most once.

  Assume now four agents $a,b,c,d$.  Then the call sequence $ab,bc,bd$
  is compliant with the protocol independently on the assumptions
  about the privacy degree and observance.  Let us analyse the
  situation after this call sequence took place.

  Assume first the privacy degree $\myprivacy{1}$.  Then agent $c$
  knows which calls took place and hence knows that after the third call
  agent $d$ is familiar with her secret, $C$. So after these three
  calls agent $c$ cannot call agent $d$ anymore.

  The situation changes when the privacy degree is $\myprivacy{2}$.
  Through the second call agent $c$ learns the secret $A$, so she
  knows that the first call was $ab$ or $ba$.  Agent $c$ is not
  involved in the third call, but by the assumed privacy degree she
  still knows that a third call has taken place.

  Assume now that the agents have common knowledge of the
  protocol. So agent $c$ knows that each pair of agents can
  communicate at most once. Hence she can conclude that $d$ must be
  involved in the third call and consequently that the third call was
  between agent $d$ and agent $a$ or $b$.  Agent $c$ therefore now
  knows that after the third call agent $d$ is familiar with at least
  $3$ secrets: $A, B, D$ if the call was with agent $a$ or
  $A, B, C, D$ if the call was with agent $b$.  But agent $c$ cannot
  anymore conclude that agent $d$ is familiar with her secret, $C$,
  and consequently can call $d$.

  Suppose now that the privacy degree is still $\myprivacy{2}$ but the
  call sequence is $ab, bc, cd, bd$. Consider now agent $a$. After the
  fourth call she knows that after the call $ab$ three calls took
  place between the agents $b,c,d$. Further, she knows that each pair
  of agents can communicate at most once. So agent $a$ concludes that
  each pair of agents from $\C{b,c,d}$ communicated precisely once. In
  particular both agents $c$ and $d$ communicated with agent $b$ and
  hence both of them are familiar with the secret $A$.  So after these
  four calls agent $a$ cannot call anymore any agent.

  Finally, consider the privacy degree $\myprivacy{3}$  and suppose the
  call sequence is $ab,bc,bd$ or $ab, bc, cd, bd$.  Then agent
  $a$ does not know whether any calls took place after the call
  $ab$. In particular she cannot conclude that any of the agents $c$
  and $d$ are familiar with her secret and hence can call either $c$
  or $d$.  \HB
\end{example}

To discuss the matters further let us be more precise about the syntax
of the protocols.  An \bfe{epistemic gossip protocol} (in short a
protocol) consists of the union of $|\Agents|$ sets of instructions,
one set for each agent.  Each instruction is of the form
\begin{center}
{\em if} $\phi$ {\em then} execute call $\Call$,
\end{center}
in symbols $\phi \to \Call$, where $\phi$ is a Boolean combination of
formulas of the form $K_a \psi$, where 
$a$ is the caller in the call $\Call$.
The formula $\phi$ is referred to as an \bfe{epistemic guard}. Such instructions
are executed iteratively, where at each time one instruction is
selected (at random, or based on some fairness considerations) whose
guard is true after the call sequence executed so far.\footnote{This
  simple rendering of protocols suffices for the purposes of this
  section. More sophisticated formalizations of epistemic gossip
  protocols have been provided in \cite{ADGH14,apt15epistemic}.}

We therefore view a protocol $P$
as a set of instructions $\phi \to \Call$. For example, the
instructions composing Protocol \ref{protocol.mineen}, are of the form
\[
\neg K_a F_bA \to ab
\]
for all agents $a$ and $b$. That is, if $a$ does not know whether $b$
is not familiar with her secret, $a$ calls $b$.

To justify the restriction on the syntax of the epistemic guards note
the following observation.

\begin{note} \label{note:p3} Consider a call type $\tau$ such that
  $\tau(\prtype) = \myprivacy{3}$.  Then for all agents $a, b, c$ and
  all call sequences $\CSequence$ and formulas $\phi$
\[
\mbox{$\CSequence \models^{\tau} K_a \phi$ iff $\CSequence. bc \models^{\tau} K_a \phi$.}
\]
Consequently, the same equivalence holds for all formulas that are
Boolean combinations of formulas of the form $K_a \phi$, so in
particular for all epistemic guards used in the instructions for agent $a$.
\end{note}

\begin{proof}
By Definition \ref{def:aux} if the privacy type of $\tau$ is $\myprivacy{3}$ then
$\CSequence \sim^{\tau}_a \CSequence. bc$, which implies the claim.
\qed
\end{proof}

This note states that the calls in which agent $a$ is not involved
have no effect on the truth of the epistemic guards used in the
instructions for agent $a$.  If we allowed in the epistemic guards for
agent $a$ as conjuncts formulas not prefixed by $K_a$, this natural
and desired property would not hold anymore.

Indeed, assume the privacy type $\myprivacy{3}$ and consider the
protocol for three agents, $a, b, c$, in which the only instructions
are $\neg F_b A \land F_b C \to ab$ for agent $a$ and
$\neg F_b C \to bc$ for agent $b$.  Then initially only the call $bc$
can be performed. After it, the call $ab$ can be performed upon which
the protocol terminates. In other words, the call $bc$, of which agent
$a$ is not aware, affects the truth of its epistemic guard, which
contradicts the idea behind the privacy type $\myprivacy{3}$.

For the privacy type $\myprivacy{1}$ this restriction on the syntax of
the epistemic guards is not needed as then all formulas are equivalent
to the propositional ones.

\begin{note} \label{not:pr1}
  Consider a call type $\tau$ such that
  $\tau(\prtype) = \myprivacy{1}$.  Then for all agents $a$ and all
  formulas $\phi$ and call sequences $\CSequence$
\[
  \mbox{$\CSequence \models^{\tau} K_a \phi$ iff
    $\CSequence \models^{\tau} \phi$.}
\]
\end{note}

\begin{proof}
  This is a direct consequence of the fact that when the privacy type
  of $\tau$ is $\myprivacy{1}$ then by Lemma \ref{pro:1}$(i)$ each relation
  $\sim^{\tau}_a$ is the identity.  \qed
\end{proof}

Let us return now to the matter of common knowledge of a protocol.  In
Definition \ref{def:aux} the $\calltype$-dependent indistinguishability
relations are constructed assuming that any call is possible after any
call sequence. This builds in the resulting gossip models
$\Model^\calltype = (\CSequences^\calltype, \set{\sim^\calltype_a}_{a
  \in \Agents})$ the assumption that agents may consider any call
sequence possible in principle, including calls that are not legal
if we assume that the agents have common knowledge of the protocol
in use.


Specifically, given a gossip model
$\Model^\calltype = (\CSequences^\calltype, \set{\sim^\calltype_a}_{a
  \in \Agents})$ and a protocol $P$ we define the
\bfe{computation tree}
$\CSequences^\calltype_P \subseteq \CSequences^\calltype$ of 
$P$ (cf.~\cite{apt15epistemic}) as the set of call sequences
inductively defined as follows:
\begin{description}
\item{[Base]} $\epsilon \in \CSequences^\calltype_P$,

\item{[Step]} If $\CSequence \in \CSequences^\calltype_P$ and
  $\CSequence \models^\calltype \phi$ then
  $\CSequence.\Call \in \CSequences^\calltype_P$, where
  $\phi \to \Call \in P$.
\end{description}
So $\CSequences^\calltype_P$ is a (possibly infinite) set of finite
call sequences that is iteratively obtained by performing a `legal'
call (according to protocol $P$) from a `legal' (according to protocol
$P$) call sequence. We refer to such legal call sequences as
$P$-compliant.\footnote{We call $\CSequences^\calltype_P$ a tree since
  its elements can be arranged in an obvious way in (a possibly
  infinite, but finitely branching) tree.}

Note however, that when building such a computation tree, the
epistemic guard $\phi$ is evaluated with respect to the underlying
gossip model $\Model^\calltype$, which may well include call sequences
that are not $P$-compliant. So in order to restrict the domain of the
gossip model to only $P$-compliant sequences, the epistemic guards of
the protocol need to be evaluated, and to do that one needs in turn a
gossip model, which contains only $P$-compliant sequences. This
circularity is not problematic for the call types involving privacy
degrees $\myprivacy{1}$ and $\myprivacy{2}$, as the $\sim_a^{\tau}$
relations then link only sequences of equal length, allowing therefore
for call sequences and these equivalence relations to be inductively
constructed in parallel. That is however not the case for the call types
involving privacy degree $\myprivacy{3}$, as then call sequences of
any length may be indistinguishable from the actual call sequence.

We propose here a solution to the above issue, showing how
under a natural assumption on the syntax of the epistemic guards one
can construct, also for the privacy degree $\myprivacy{3}$, a gossip
model which consists only of call sequences that are compliant with a given
protocol $P$.

Fix till the end of the section an arbitrary call type $\tau$.  First,
we introduce the definition of semantics relativised to a set
$X \sse \CSequences^\tau$ of call sequences. Let
$\Model_X^\calltype = (X, \set{\sim^\calltype_a}_{a \in \Agents})$,
where each $\sim^\calltype_a$ relation is restricted to $X \times X$,
and let $\CSequence \in X$.  Then the definition of semantics is the same
as before with the except of the formulas of the form $K_a \phi$:
\begin{eqnarray*}
  (\Model_X^\calltype, \CSequence) \models  K_a \phi &  \mbox{iff}  & \forall \CSequenced \in X \mbox{ such that } \CSequence \sim^{\calltype}_a \CSequenced, ~ (\Model_X^\calltype, \CSequenced) \models \phi.
\end{eqnarray*}

Fix now a protocol $P$ and a set $X \sse \CSequences^\tau$. We define
the relativised computation tree of $P$ as the set
$\CSequences^\calltype_{(P,X)}$ obtained by replacing the above Base and Step
conditions by
\begin{description}
\item{[Base]} $\epsilon \in \CSequences^\calltype_{(P,X)}$,

\item{[Step]} If $\CSequence \in X \cap \CSequences^\calltype_{(P,X)}$
  and $(\Model_X^\calltype, \CSequence) \models \phi$ then
  $\CSequence.\Call \in \CSequences^\calltype_{(P,X)}$, where
  $\phi \to \Call \in P$,
\end{description}
and refer to each call sequence from
$\CSequences^\calltype_{(P,X)}$ as \emph{$(P,X)$-compliant}.

We now limit the syntax of epistemic guards as follows.  A
formula $\hat{K}_a\phi$ is an abbreviation for $\neg K_a \neg \phi$
and $\hat{\lang}$ denotes the existential fragment of $\lang$,
consisting of only literals, $\vee$, $\land$, and $\hat{K}_a$.


The following lemma clarifies the introduction of the language $\hat{\lang}$.

\begin{lemma} \label{lem:mono}
If $X \sse Y \sse \CSequences^\tau$ then
\[
\mbox{for all formulas $\phi \in \hat{\lang}$ and all  $\CSequence \in X$, 
$(\Model_X^\calltype, \CSequence) \models  \phi \implies
(\Model_Y^\calltype, \CSequence) \models  \phi$.}
\]
\end{lemma}
\begin{proof}
The only case that  requires explanation is when $\phi$ is
of the form $\hat{K}_a \psi$. Suppose that
$(\Model_X^\calltype, \CSequence) \models  \phi$. Then for some $\CSequenced \in X$
such that $\CSequence \sim^{\calltype}_a \CSequenced$, 
$(\Model_X^\calltype, \CSequenced) \models  \psi$. By the induction hypothesis
$(\Model_Y^\calltype, \CSequenced) \models  \psi$, so by definition
$(\Model_Y^\calltype, \CSequence) \models  \phi$. 
\qed
\end{proof}

Define next an operator
$\rho^P: 2^{\CSequences^\tau} \to 2^{\CSequences^\tau}$ by
\[
\rho^P(X) = X \cap \CSequences^\calltype_{(P,X)}.
\]
That is, $\rho^P$ removes from a given set $X$ of call sequences those
that are not $(P,X)$-compliant. What we are after is a set from which no
sequences would be removed, so a fixpoint of $\rho^P$. 

\begin{proposition} \label{prop:CK} Suppose the epistemic guards
  of a protocol $P$ are all from $\hat{\lang}$. Then there exists an
  $X \subseteq \CSequences^\tau$ such that $X = \rho^P(X)$.
\end{proposition}
\begin{proof}
  Suppose that $X \subseteq Y$ and
  $\CSequence \in X \cap \CSequences^\calltype_{(P,X)}$.  We prove that
  $\CSequence \in \CSequences^\calltype_{(P,Y)}$ by induction on the
  length of $\CSequence$.  If $\CSequence = \epsilon$, then
  $\CSequence \in \CSequences^\calltype_{(P,Y)}$ by the Base condition.

  Otherwise, by the Step condition $\CSequence$ is of the form
  $\CSequence'.\Call$, where
  $\CSequence' \in X \cap \CSequences^\calltype_{(P,X)}$, and for some
  $\phi \in \hat{\lang}$,
  $(\Model_X^\calltype, \CSequence') \models \phi$, and
  $\phi \to \Call \in P$.  By the induction hypothesis
  $\CSequence' \in \CSequences^\calltype_{(P,Y)}$.  Further,
  $\CSequence' \in Y$ and by Lemma \ref{lem:mono}
  $(\Model_Y^\calltype, \CSequence') \models \phi$, so
 $\CSequence \in \CSequences^\calltype_{(P,Y)}$. 

 It follows that $\rho^P$ is a monotonic function, that is,
 $X \subseteq Y$ implies $\rho^P(X) \subseteq \rho^P(Y)$.  By the
 Knaster-Tarski theorem of \cite{Tar55} $\rho^P$ has therefore
 fixpoints, including a largest and a smallest one.  \HB
\end{proof}

Intuitively, when the domain $X \sse \CSequences^\calltype$ of a
gossip model is a fixpoint of $\rho^P$, then the restriction of the
definition of the indistinguishability relations $\sim^\tau_a$ to such
a domain has the effect that the call sequences considered possible by the
agents coincide with the call sequences generated by the protocol.
Such gossip models incorporate then the assumption that there is
common knowledge among the agents about the protocol in use.

Furthermore, by the Knaster-Tarski theorem one can construct the
largest fixpoint of $\rho^P$ by iteratively applying $\rho^P$ to
$\CSequences^\tau$. Such fixpoint $\nu \rho^P$ is the most natural
domain for a gossip model that realises the assumption of common
knowledge of the protocol, with the ($P, \nu \rho^P)$-compliant call
sequences viewed as the $P$-compliant ones.

When the privacy degree is $\myprivacy{1}$ such a gossip model
has a very simple structure, namely
$(\CSequences_P^\calltype, \set{\sim^\calltype_a}_{a \in \Agents})$.

\begin{corollary}
  Consider a protocol $P$ and a call type $\tau$ such that
  $\tau(\prtype) = \myprivacy{1}$.  Then 
  $\nu \rho^P = \CSequences_P^\tau$.
\end{corollary}
\begin{proof}
  Note that we always have
  $\rho^P(\CSequences^\tau) = \CSequences_P^\tau$.  We now show that
  $\CSequences_P^\tau \sse
  \CSequences^\calltype_{(P,\CSequences_P^\tau)}$ by induction on the
  length of the call sequences. We only need to consider the induction
  step.  So consider some $\cc.\Call \in \CSequences_P^\tau$. By
  definition $\cc \in \CSequences_P^\tau$ and
  $\CSequence \models^\calltype \phi$, where $\phi \to \Call \in P$,
  and by the induction hypothesis
  $\cc \in \CSequences^\calltype_{(P,\CSequences_P^\tau)}$.

  Let $\phi'$ be obtained from $\phi$ by removing all occurrences of
  $K_a$ for all agents $a$.  By Note \ref{not:pr1} relativised to an
  arbitrary $X \sse \CSequences^\tau$ such that $\cc \in X$ we have
  $\CSequence \models^\calltype \phi$ iff
  $\CSequence \models^\calltype \phi'$ iff
  $(\Model_X^\calltype, \CSequence) \models \phi'$ iff
  $(\Model_X^\calltype, \CSequence) \models \phi$.  So in particular
  $(\Model_{\CSequences_P^\tau}^\calltype, \CSequence) \models \phi$
  and hence by definition
  $\cc.\Call \in \CSequences^\calltype_{(P,\CSequences_P^\tau)}$.

Consequently $\rho^P(\CSequences_P^\tau) = \CSequences_P^\tau \cap
\CSequences^\calltype_{(P,\CSequences_P^\tau)} = \CSequences_P^\tau$ and hence
$\CSequences_P^\tau$ is the largest fixpoint of $\rho^P$.
\qed
\end{proof}

The syntactic restriction on the epistemic guards used in Proposition
\ref{prop:CK} is clearly satisfied by Protocol \ref{protocol.mineen}
as its guards can be rewritten as $\hat{K}_i \neg F_j I$.  The same is
the case for all protocols studied in \cite{apt15epistemic} since by
Corollary \ref{cor:Ka} for all call types and all agents $a$ and $b$
the formulas $K_a F_a B$ and $F_a B$ are equivalent and consequently
each formula $F_a B$ can be replaced by $\neg \hat{K}_a \neg F_a B$.


\section{Conclusions}
\label{sec:conclusions}

We provided an in-depth study of 18 different types of communication
relevant for epistemic gossip protocols and modelled their epistemic
effects in a uniform way through different indistinguishability
relations. This led us to establish a precise map of the relative
informativeness of these types of communication (Theorem
\ref{thm:classification}). In turn, this result allowed us to prove
general results concerning the epistemic effects of call sequences
under different communication regimes (Propositions \ref{prop:1} and
\ref{prop:newnew}) and to advance a natural proposal on how to model and
analyse agents' common knowledge of gossip protocols (Proposition
\ref{prop:CK}), a still under-investigated issue in the literature.

Several natural directions for future research present themselves. We
mention three of them. The first question concerns the axiomatisation of
the modal language $\lang$ introduced in Section
\ref{sec:language}. This problem is parametrised by the underlying
indistinguishability relations introduced in Section
\ref{sec:relations}. For example, by Note \ref{not:pr1} the
equivalence $\phi \leftrightarrow K_a \phi$ holds for the privacy type
$\myprivacy{1}$ but not for the other two.

Actually, even the axiomatization of the $F_a S$
formulas is not straightforward, as it has to take into account the
nature of the communication.  Indeed, consider the following formula,
where $a \neq b$:
\[                                                                                              
\Big(F_b A \land \bigwedge_{i \neq a, b} \neg F_i A\Big) \to F_a B.                      
\]

It states that if agent $b$ is the only agent (different from $a$) familiar
with the secret of $a$, then agent $a$ is familiar with the secret of $b$.
A more general version is:
\[                                                                                              
\Big(\bigvee_{i \in X} F_i A \land \bigwedge_{i \not\in X \cup \{a\}} \neg F_i A\Big) \to 
\bigvee_{i \in X} F_a I,                                                                        
\]
where $a \not\in X$.

Intuitively it states that if somebody from a group $X$, to which $a$
does not belong, is familiar with her secret and nobody from outside
of the group $X$ (except $a$) is familiar with this secret, then agent
$a$ is familiar with a secret of somebody from the group $X$.
Clearly, both formulas are valid for the $\pushpull$ direction
type. 

In general such an axiomatisation project could be carried out
at several levels (cf.~\cite{gattinger18new}): by considering $F_i S$
formulas as primitive, as we did in this paper; or analysing them as
``knowing whether'' formulas (in epistemic logic notation,
$K_i S \vee K_i \neg S$) as in \cite{ADGH14}. Whether the
latter level of analysis can be easily reconciled with the one
proposed in this paper is an interesting open problem.

The second question addresses the problem of decidability of the 18
definitions of truth we introduced. In the terminology of this paper
\cite{AW18} established for the call type
$(\myprivacy{3}, \pushpull, \alpha)$ that the semantics and the
definition of truth are both decidable for the formulas without nested
modalities. It would be interesting to establish analogous results for
the remaining call types, ideally by providing a single, uniform proof
that generalises the arguments of \cite{AW18}.

The final question concerns the robustness of our analysis, and
specifically of the relationships identified in Theorem
\ref{thm:classification}, with respect to modes of gossip that involve
the transfer of higher-order epistemic information as introduced and
studied in \cite{herzig15how,herzig17how}. Intuitively, we would
expect this type of higher-order epistemic communication to have an
impact on the effects of the asymmetric communication types $\push$
and $\pull$ and for the full privacy \myprivacy{3} degree.

Finally, one could envisage other aspects of a call not considered in
this framework.  For example in \cite{ADGH14} yet another notion of
privacy was considered, according to which given a call $ab$ every
agent $c \neq a,b$ noted that at most one call took place. Then for
agent $c$ the call sequences $\epsilon$ and $ab$ are equivalent but
$\epsilon$ and $ab, ab$ are not.  Another possibility could be to
consider a notion of privacy that is intermediate between
\myprivacy{1} and \myprivacy{2}, according to which the caller is
anonymous but the callee not. Then for agent $c$ the call sequences
$ab$ and $ad$ are equivalent but $ab$ and $bd$ are not.

\section*{Acknowledgments}
We thank Hans van Ditmarsch for most useful and extensive discussions
on the subject of this paper. We are also grateful to anonymous
referees of this and earlier versions of this paper for helpful
comments.  The first author was partially supported by the NCN grant
nr 2014/13/B/ST6/01807.


\bibliographystyle{plain}  
\bibliography{gossip_biblio}

\end{document}